\documentclass[twocolumn]{autartSofie}  
\usepackage{amsmath,amsfonts,amssymb,mathtools}
  \usepackage{graphicx,subcaption}
  \usepackage{graphics,sidecap}

\usepackage{tikz}
\usetikzlibrary{decorations.shapes,shapes.geometric,shadows,arrows,automata,positioning,calendar,mindmap,backgrounds,scopes,chains,er,patterns}
\usepackage{color} 
\usepackage{pgfplots}
\usepackage{paralist}
 
\newcommand{\citep}[1]{\cite{#1}}
\newenvironment{proof}{\textbf{Proof}%
}{%
\hfill$\Box$
}
 \newcommand{\alphabeth}{{\Sigma}}
 \newcommand{\Var}{{R}}
 \newcommand{\var}{{\sigma}}
 \newcommand{\letter}{{\alpha}}
	\newtheorem{problem}{Problem}
\newtheorem{proposition}[thm]{Proposition}

\newtheorem{definition}{Definition}

\newcommand{\Real}{\mathbb R}        

\newcommand{\s}{\mathbf{S}}    
\newcommand{\Mset}{\mathcal{G}} 
\newcommand{\M}{\mathbf{M}}    

\newcommand{\X}{\mathbb{X}}     
    \newcommand{\x}[1]{\mathbf{x}(#1)} 
    \newcommand{\xT}{\mathbf{x}^T} 

\newcommand{\A}{\mathbb{U}}     
    \newcommand{\ac}[1]{u(#1)}         
\newcommand{\Y}{\mathbb{Y}}     
    \newcommand{\y}[1]{\hat{y}(#1)}          
    \newcommand{\ym}[1]{\tilde{y}(#1)}   
    \newcommand{\yo}[1]{y_0(#1)}       

\newcommand{\BLTL}{\psi}        
\newcommand{\ap}{{p_i}}            

\newcommand{\always}{\Box}

\newcommand{\until}{\mathbin{\sf U}}

\newcommand{\nex}{\mathord{\bigcirc}}

\newcommand{\Nsample}{{N_s}}       
\newcommand{\pa}{\theta}        
\newcommand{\parTrue}{\pa^0}               

\newcommand{\p}[1]{\mathbf{P}\left(#1\right)}     
\newcommand{\pd}[1]{p\left(#1\right)}     

\usepackage{caption}
\allowdisplaybreaks
    \makeatletter
    \def\ps@pprintTitle{%
      \let\@oddhead\@empty
      \let\@evenhead\@empty
      \def\@oddfoot{\reset@font\hfil\thepage\hfil}
      \let\@evenfoot\@oddfoot
    }
    \makeatother
\begin{document} 

\begin{frontmatter}
\title{Data-driven and Model-based Verification:\\ 
a Bayesian Identification Approach}

\author[TUE]{Sofie Haesaert}\ead{s.haesaert@tue.nl},      
\author[TUE]{Paul M.J. Van den Hof}\ead{P.M.J.Vandenhof@tue.nl},  
\author[OX]{Alessandro Abate}\ead{alessandro.abate@cs.ox.ac.uk}       
               \address[TUE]{
                Department of Electrical Engineering, 
       Eindhoven University of Technology, 
       Eindhoven, The Netherlands}%
       \address[OX]{Department of Computer Science, 
       University of Oxford, 
       Oxford, United Kingdom
       }
\begin{abstract}
This work develops a measurement-driven and model-based 
formal verification approach, applicable to systems with partly unknown dynamics. 
We provide a principled method, 
grounded on reachability analysis and on Bayesian inference, 
to compute the confidence that a physical system driven by external inputs and accessed under noisy measurements, 
verifies a temporal logic property.  
A case study is discussed, 
where we investigate the bounded- and unbounded-time safety of a partly unknown linear time invariant system. 
\end{abstract}
\end{frontmatter}
 
\maketitle
\section{Introduction}\label{Introduction}  
 The design of complex, high-tech, safety-critical systems such as autonomous vehicles, intelligent robots, 
and cyber-physical infrastructures, demands guarantees on their correct and reliable behaviour.  
Correct functioning and reliability over models of systems can be attained by the use of formal methods.  
Within the computer sciences, the formal verification of software and hardware has successfully led to industrially relevant and impactful applications \cite{Clarke2008}. 
Carrying the promise of a decrease in design faults and implementation errors and of correct-by-design synthesis,  
the use of formal methods, such as model checking \cite{Clarke2008}, has become a standard in the avionics, automotive, and railway industries \cite{Vardi2006}. 
Life sciences \cite{belta2002control,del2009engineering} and general engineering applications \cite{Belta2007,burdick} have also recently pursued the extension of these successful techniques from the computer science:  
this has required a shift from finite-state to physical and cyber-physical models that are of practical use in nowadays science and technology \cite{lee08,Tabuada2009}.  

The strength of formal techniques, such as model checking, 
is bound to the fundamental requirement of having access to a given model,
obtained from the knowledge of the behaviour of the underlying system of interest.
In practice, for most physical systems the dynamical behaviour is known only in part:  
this holds in particular with biological systems \cite{AHW12}   
or with classes of engineered systems where, as a consequence, 
the use of uncertain control models built from data is a common practice \cite{hjalmarsson2005experiment}. 
 
 
Only limited work within the formal methods community deals with the verification of models with partly unknown dynamics. 
Classical results \cite{Batt2007,Henzinger1996} consider the verification problem for non-stochastic models described by differential equations and with bounded parametric uncertainty. 
Similarly, but for continuous time probabilistic models, \cite{Bortolussi2014,Ceska} explore the parameter space with the objective of model verification (respectively statistical or probabilistic).  
Whenever full state measurements of the system are available, 
Statistical Model Checking (SMC) \cite{sen2004statistical,Legay} 
replaces model(-based) checking procedures with empirical testing of formalised properties.
SMC is limited to fully observable  stochastic systems with little or no non-determinism, 
and may require the gathering a large set of measurements.  
Extensions towards the inclusion of non-determinism have been studied in \cite{Henriques2012,Legay2013}, 
with preliminary steps towards Markov decision processes. 
Related to SMC techniques, but bound to finite state models, 
\cite{Chen2012,Mao2012a,Sen2004} assume that the system is encompassed by a finite-state Markov chain and efficiently use data 
to learn the corresponding model and to verify it. 
 Similarly, 
\cite{Bartocci2013,bortolussi:learning} employ machine learning techniques to infer finite-state Markov models from data over specific logical formulae.   
 
An alternative approach, 
allowing both partly unknown dynamics over uncountable (continuous) variables and noisy output measurements, 
is the usage of a Bayesian framework relating the confidence in a formal property to the uncertainty of a model built from data. 
When applied on nonlinearly parameterised linear time invariant (LTI) models 
this approach introduces huge computational problems, which 
as proposed in \cite{Gyori2014}, can only be mitigated by statistical methods.  
Instead, to obtain reliable and numerical solutions, we propose the use of linearly parameterised model sets defined through orthonormal basis functions to represent these partially unknown systems. This is a broadly used framework in system identification \cite{heuberger2005modelling,hjalmarsson2005experiment}:  
it allows for the incorporation of prior knowledge, while maintaining the benefits (computational aspects) of linear parameterisations. 
Practically, it has been widely used for the modelling of physical systems, such as the thermal dynamics of buildings \cite{Virk1994}.  
In contrast,  in this paper we  pursue a promising new numerical approach: 
instead of employing directly a nonlinearly parameterised model, 
we embed it in a linearly parameterised one 
via a series expansion of orthonormal basis functions.  

In this contribution we further analyse and extend the related results in \cite{ACCSofie}, obtained for a time-bounded subset of temporal logic properties, to unbounded-time temporal logic properties, and analyse their robustness.

\section{General Framework and Problem Statement} \label{sec:2}

In this section, we provide a novel methodology to verify whether a system $\s$
satisfies a specification 
$\psi$, formulated in a suitable temporal logic, 
by integrating the partial knowledge of the system dynamics with data obtained from a measurement set-up around the system. 

Let us further clarify this framework. 
Let us denote with $\s$ a physical system, 
or equivalently the associated dynamical behaviour.   
A signal input  $\ac{t}\in \A, t\in \mathbb N$, captures how the environment acts on the system.  
Similarly, an output signal $\yo{t}\in \Y$ indicates how the system interacts with the environment, 
or alternatively how the system can be measured.  
Note that the input and output signals are assumed to take values over continuous domains.
The system dynamics can be described via mathematical models, 
which express the behavioural relation between its inputs and outputs.  The knowledge of the behaviour of the system is often limited or uncertain, making it impossible to analyse its behaviour via that of a ``true'' model.  
In this case,  
a-priori available knowledge allows to construct a model set $\Mset$ with elements $\M\in\Mset$: 
this model class supports the structured uncertainty as a 
distribution over a parameterisation $\pa\in\Theta$, 
$\Mset=\{\M(\pa)|\pa\in\Theta\}$. 
The unknown ``true'' model $\M(\pa^0)$ representing $\s$, 
is assumed to be an element of $\Mset$, 
namely $ \pa^0\in \Theta$: 
as an example, model sets $\Mset$ obtained through first principles adhere to this classical assumption.

Samples can be drawn from the underlying physical system via a measurement set-up, 
as depicted in Figure \ref{fig:System}.  
An experiment consists of a finite number ($\Nsample$) of input-output samples drawn from the system, 
and is denoted by $Z^\Nsample=\{\ac{t}_{ex},\ym{t}_{ex}\}_{t=1}^\Nsample$,
where $\ac{t}_{ex}\in \A$  is the input for the experiment and  $\ym{t}_{ex}$ is a (possibly noisy) measurement of $\yo{t}_{ex}$. 
In general, the measurement noise can enter non-additively and be a realisation of a stationary stochastic process.\footnote{Both the operating conditions of the experiment, that is the input signal $\ac{t}_{ex}$ and the initial state $\x{0}_{ex}$, and the measurements have been indexed with $ex$ to distinguish them from the operating conditions of interest for verification, to be discussed shortly.}   
We assume that at the beginning of the measurement procedure (say at $t=0$), 
the initial condition of the system, encompassed by the initial state $\x{0}_{ex}$ of models in $\M$, 
is either known, or, when not known, has a structured uncertainty distribution based on the knowledge of past inputs and/or outputs.  
As reasonable, we implicitly consider only well-defined problems, 
such that for any model representing the system,  
given a signal input $\ac{t}_{ex}$ and an (uncertainty distribution for) $\x{0}_{ex}$, 
the probability density distribution of the measured signal can be fully characterised. 
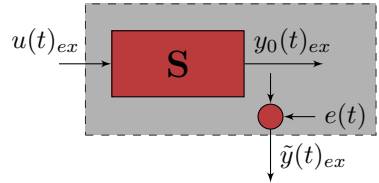
\begin{figure}[h]
 \begin{minipage}[b]{\columnwidth}\definecolor{tuecyan}{RGB}{0,0,173}
\definecolor{tuewarmred}{RGB}{170,10,13}
\definecolor{tueblue}{RGB}{10,10,173}
\centering
    \tikzstyle{block} = [draw, fill=tuewarmred!80, rectangle,
    minimum height=2.5em, minimum width=5em]
    \tikzstyle{sum} = [draw, fill=tuewarmred!80, circle, node distance=3em]
    \tikzstyle{input} = [coordinate]
    \tikzstyle{output} = [coordinate]
    \tikzstyle{pinstyle} = [pin edge={to-,thin,black}]
\begin{tikzpicture}[auto,  node distance=3em,>=latex',scale=1]
\begin{scope}[overlay]
    \node [input,name=input,scale=2,label={\textcolor{black}{ $\ac{t}_{ex}$ \ \ }}] { };
    \node [block, right of=input, node distance=4.5em] (system) {\large \(\s\)};
    \node [right of=system,node distance=3.5em] (split) {};
    \node [output, right of=split, node distance=2em] (out1) { };
    \node [sum, below of=split,node distance=2em ] (sumNoise) {};
    \node[name=et,node distance=1 cm,right of= sumNoise]{$e(t)$};
    \node [output, below of=sumNoise, node distance=2.5em] (out2) { };
    \draw [->] (input) -- node[name=u] {} (system);
    \draw [->] (system) -- node [name=y] {$\ \ \yo{t}_{ex}$}(out1);
    \draw [->] (split) -- node [name=y] {}(sumNoise);
    \draw [->] (sumNoise) -- node [name=y] {$\ym{t}_{ex}$}(out2);
    \draw [->] (et) -- node [name=y] {}(sumNoise);
\end{scope}
   \begin{pgfonlayer}{background}
   \path (input.west |- system.north)+(+1em,+1em) node (a) {};
   \path (sumNoise.east |- out2.north)+(3.5em,+1.8em) node (c) {};
   \draw[dashed,black!80,fill=black!30]
            (a) rectangle (c);
\end{pgfonlayer}
\end{tikzpicture}
\end{minipage}  
     \caption{System and measurement setup. 
In the measurement setup (grey box) the measured output $\ym{t}_{ex}$ includes the system output $\yo{t}_{ex}$ and the measurement noise $e(t)$.  
Data collected from experiments comprises the input $\ac{t}_{ex}$ and the measured output $\ym{t}_{ex}$ signals.   
 }\label{fig:System} \end{figure} 
\\*
The end objective is to analyse the behaviour of system $\s$. 
We consider properties encoded as specifications $\BLTL$ and expressed in a temporal logic of choice (to be detailed shortly).   
Let us remark that the behaviour of $\s$ to be analysed is bound to a set of operating conditions that are pertinent to the verification problem and that will be indexed with $ver$:  
this comprises the set of possible input signals $\ac{t}_{ver}$ (e.g., a white or coloured noise signal, or a non-deterministic signal $\ac{t}_{ver}\in \A_{ver}\subseteq \A$
 ), and of the set of initial states $\x{0}_{ver}\in\X_{ver}$ for the mathematical models $\M$ reflecting past inputs and/or outputs of the system.
The system satisfies a property if the ``true'' model representing it satisfies it, namely $\s\vDash\BLTL$ if and only if $\M(\pa^0)\vDash\BLTL$.  

In this work we consider the satisfaction of a property $\M(\pa)\vDash\BLTL$ as a \emph{binary-valued mapping} from the parameter space $\Theta$.  
More generally, 
when in addition to the measurements of the system also its transitions are disturbed by stochastic noise, 
then property satisfaction  
is a mapping from the parameter space $\Theta$ to the interval $[0,1]$, 
and quantifies the probability that the model $\M(\pa)$ satisfies the property. 
This mapping generalises the definition of the satisfaction function introduced in \cite{Bortolussi2014}, 
and is now stated as follows. 
\begin{definition}[Satisfaction Function] 
Let $\Mset$ be a set of models $\M$ that is indexed by a parameter $\pa\in \Theta$, 
and let $\psi$ be a formula in a suitable temporal logic. 
The satisfaction function $f_\psi:\Theta\rightarrow [0,1]$ associated with $\psi$ is 
\begin{equation}\label{eq:sat}
f_\psi(\theta) = \mathbf{P}\left( \M(\pa)\vDash \psi\right). 
\end{equation} 
\end{definition}
 Let us assume that the satisfaction function $f_\psi$ is measurable and entails a decidable verification problem (e.g., a model checking procedure) for all $\pa\in \Theta$. 

\smallskip

\begin{problem}\label{prob_stat}
\textit{
For a partly unknown physical system $\s$, 
under prior knowledge on the system given as a parameterised model class $\Mset$ supporting an uncertainty distribution over the parameterisation, 
gather possibly noisy data drawn from the measurement setup and verify properties on $\s$ expressed in a temporal logic of choice, 
with a formal quantification of the confidence of the assertion. 
} 
\end{problem}

\subsection{A Bayesian Framework for Data-driven Modelling and Verification}\label{sec:Bayesian}
 
Consider Problem \ref{prob_stat}. 
Denote loosely with $\p{\cdot}$ and $\pd{\cdot}$ respectively a probability measure and a probability density function, 
both defined over a continuous domain.    
We employ Bayesian probability calculus \cite{Lindley2011} 
to express the confidence in a property as a measure of the uncertainty distribution 
defined 
 the set $\Mset$.  
By adopting the Bayesian framework, 
uncertainty distributions are handled 
as probability distributions of random variables. 
Therefore the confidence in a property is computed as a probability measure $\p{\cdot}$ via the densities $\pd{\cdot}$ over the uncertain variables.

\begin{proposition}[Bayesian Confidence]\label{thm1}
Given a specification $\BLTL$ and a data set $Z^\Nsample$, the confidence that $\s\vDash \BLTL$  can be quantified via inference as
\begin{equation}\label{eq:BayesianConf}
  \textstyle \p{ \s\vDash\BLTL \mid Z^{N_s}} = \int_{\Theta}f_\BLTL(\pa) 
   \pd{\pa|Z^{N_s}}d\pa\ .
\end{equation}
 where $f_\BLTL$ is the satisfaction function given in \eqref{eq:sat}.   
 The \emph{a-posteriori} uncertainty distribution $\pd{\pa|Z^{N_s}}$, given the data set $Z^{N_s}$, is based on parametric inference over $\pa$ as
\begin{equation}\label{BayesianID}
  \textstyle\pd{\pa|Z^{N_s}} = \frac{\pd{Z^{N_s}|\pa}\pd{\pa}}{\int_{\Theta}\pd{Z^{N_s}|\pa}\pd{\pa}d \pa}\ ,
\end{equation} 
which presumes an uncertainty distribution $\pd{\pa}$ over the parameter set $\Theta$, 
representing  the prior knowledge. 
\end{proposition}

The statement can be formally derived based on standard Bayesian calculus, as in \cite{Lindley2011}.  
 We have chosen to employ a Bayesian framework, as per \eqref{BayesianID}, 
since it allows to reason explicitly over the uncertain knowledge on the system and to work with the data acquired from the measurement setup.   
This leads to the efficient incorporation of the available knowledge and to its combination with the data acquisition procedure, 
in order to compute the confidence on the validity of a given specification over the underlying system.    
As a special instance, 
this result can be employed for Bayesian hypothesis testing \cite{Zuliani2013a}.   
 As long as the mapping $f_\BLTL$ is measurable, 
the models in the model set (and hence the system represented by it) can be characterised by either probabilistic or non-probabilistic dynamics. 
\begin{rem}
 In statistical model checking \cite{Legay,sen2004statistical}, 
the objective is to replace the computationally tolling verification of a system over bounded-time properties by the empirical (statistical) testing of the relevant specifications over finite executions drawn from the system.    
In contrast, our problem statement tackles the problem of efficiently incorporating data with prior knowledge, for the formal (deductive) verification of the behaviour of a system with partly unknown dynamics -- as such our overall verification approach is, as claimed, both data-driven and model-based. 
Moreover, by separating the operational conditions in an experiment from those of importance for the verification procedure, 
the system can be verified over non-deterministic inputs, 
encompassing as such both controller and disturbance inputs, or modelling errors.%
\end{rem}  
\subsection{Computational Approaches}
 The Bayesian approach is widely applicable to different types of properties and models,  
however its computational complexity might in practice limit its implementation.    
 In the literature the satisfaction function is related to the exploration of a parameter set over the validity of a formal property $f_\psi(\pa)$, and has been studied 
for autonomous models in continuous time in  
\cite{Batt2007,Frehse2008,Henzinger1996}. Analytical solutions to the parametric inference equation \eqref{BayesianID}  can be found if the prior is a conjugate distribution. For linear dynamical systems, closed-form solutions are given inter alia in \cite{Peterka1981a}.  
In general \eqref{eq:BayesianConf}-\eqref{BayesianID} in Proposition \ref{thm1} lack analytical solutions,  
and the assessment of the satisfaction function \eqref{eq:sat} may be computationally intensive. 
Statistical methods such as the one proposed in \cite{Gyori2014} on a similar Bayesian approach lead to involved computations and introduce additional uncertainty from Monte Carlo techniques. 

On the contrary, 
in the next section, we propose a novel computational approach over discrete-time linear time-invariant systems. 
By exploiting linear parameterisations analytical solutions of both the parametric inference and the satisfaction function are characterised for properties expressed within a fragment of a temporal logic.   

\section{LTL Verification of LTI systems}\label{sec:3} 

Consider a system $\s$ that can be represented by 
a class of finite-dimensional dynamical models that evolve in discrete-time, 
and are linear, time-invariant (LTI), and not probabilistic.   
These models depend on input and output signals ranging over $\mathbb{R}^m$ and $\mathbb R^p$, respectively, 
and on variables $\mathbf{x}_\s(t)$ taking values in 
an Euclidean space, $\mathbf{x}_\s(t)\in\X\subseteq\mathbb{R}^n$, where $n$, the state dimension, is the model order.  
The behaviour of such a system is encompassed by state-space models $(A_\s,B_\s,C_\s,D_\s)$ as  
\begin{equation} 
\s:\quad \left\{ \begin{array}{ll} \mathbf{x}_\s(t+1)&=A_\s\mathbf{x}_\s(t)+B_\s u(t),\\\yo{t}&=C_\s\mathbf{x}_\s(t)+D_\s u(t),\end{array}\right.
\end{equation}
 \noindent%
where matrices $A_\s, B_\s, C_\s, D_\s$ are of appropriate dimensions.   
Let us remark that LTI systems represent the most common modelling framework in control theory,   
a key framework leading towards generalisations to more complicated (e.g., nonlinear) dynamical models.  
The experimental measurement setup, 
as depicted in Figure \ref{fig:System}, 
consists of the signals $\ac{t}_{ex}$ and $\ym{t}_{ex}=\yo{t}_{ex}+e(t)$, 
representing the inputs and the measured outputs, respectively,  
and where $e(t)$ is an additive zero-mean, white, Gaussian-distributed measurement noise with covariance $\Sigma_e$ that is uncorrelated from the inputs.  
$\Nsample$ samples are collected within a data set $Z^{\Nsample}=\{\ac{t}_{ex},\ym{t}_{ex}\}_{t=1}^\Nsample$.

System properties are expressed, over 
a finite set of atomic propositions $p_i\in AP$, $i=1,\ldots,|AP|$, in 
Linear-time Temporal Logic \cite{Baier2008}. 
LTL formulae are built recursively via the syntax \(\psi::= \operatorname{true}\mid p \mid \neg \psi\mid \psi\wedge \psi \mid \psi\vee \psi  \mid \nex \psi \mid \psi \until \psi.\)
Let $\pi=\pi(0),\pi(1),\pi(2),\ldots \in\alphabeth^{\mathbb{N}^+}$ be a string composed of letters from the alphabet $\alphabeth = 2^{AP}$, 
and let $\pi_t=\pi(t),\pi(t+1),\pi(t+2),\ldots  $ be a subsequence of $\pi$, 
then the satisfaction relation between $\pi$ and $\psi$ is denoted as \(\pi\vDash\psi\)  
(or equivalently \(\pi_0\vDash\psi\)). 
The semantics  for the satisfaction are defined recursively over $\pi_t$ and the LTL syntax as 
\begin{equation*}
\begin{aligned}
\mbox{ {\scriptsize (true)}}\  &\pi_t\vDash \operatorname{true} & \Leftrightarrow &\operatorname{true}\\
 \mbox{ {\scriptsize (atomic prop.)}} \ &\pi_t\vDash p& \Leftrightarrow & \ p\in\pi(t)\\
 \mbox{ {\scriptsize (negation)}}\ &\pi_t\vDash \neg \psi & \Leftrightarrow &\ \pi_t\not\vDash  \psi\\
 \mbox{ {\scriptsize (conjunction)}}\ &\pi_t\vDash \BLTL_1\wedge\BLTL_2
 & \Leftrightarrow & \  \pi_t\vDash \BLTL_1 \mbox{ and } \pi_t\vDash \BLTL_2 \\
 \mbox{ {\scriptsize (disjunction)}}\ &\pi_t\vDash \BLTL_1\vee\BLTL_2
 & \Leftrightarrow & \  \pi_t\vDash \BLTL_1 \mbox{ or } \pi_t\vDash \BLTL_2 \\
 \mbox{ {\scriptsize (next)}} \ &\pi_t\vDash\nex\BLTL    & \Leftrightarrow & \  \pi_{t+1}\vDash \BLTL\\
 \mbox{ {\scriptsize (until)}} \  & \pi_t\vDash \BLTL_1\until\BLTL_2
 & \Leftrightarrow & \  \exists\, i\in\mathbb{N}:  \pi_{t+i}\vDash \BLTL_2,\\ &&&\qquad
\mbox{and } 
\forall j \in\mathbb{N}:\\ &&&\qquad 0\leq j<i, \pi_{t+j}\vDash \BLTL_1
\end{aligned}
\end{equation*}
Denote the $k$-bounded and unbounded invariance operator as $\always^k\psi=\bigwedge_{i=0}^{k}\nex^i \psi$ and $\always\psi=\neg(\texttt{true}\until\neg \psi)$, respectively.

Of interest are formal properties encoded on the input-output behaviour of the system, 
and over a time horizon $t\geq0$.  
The output $\yo{t}_{ver}\in\Y$ is labeled by a map $L:\Y\rightarrow\alphabeth$,  
which assigns {letters $\letter$} in the alphabet $\alphabeth$ {via} half spaces on the output, as  
\begin{equation}\label{eq:labelling}\textstyle
  L(y_0(t)_{ver})=\letter\in\alphabeth\ \textstyle\Leftrightarrow \ \bigwedge_{p_i\in \letter}A_{p_i} \yo{t}_{ver}\leq b_{p_i},
\end{equation}  
for given $A_{p_i}\in \mathbb{R}^{1\times p} ,\ b_{p_i}\in \mathbb{R}$ that is, 
sets of atomic propositions are associated to polyhedra over $\Y\subset \mathbb{R}^p$. 
Let us underline that properties are defined over the behaviour $\yo{t}_{ver}$ of the system, and not over the noisy measurements $\ym{t}_{ex}$ of the system in the measurement setup. 
Additionally, for the verification problem the input signal is modelled as a bounded signal $\ac{t}\in\A_{ver}$, and represents possible external non-determinism of the environment acting on the system.   
 
 \subsection{Model Set Selection}
As a first step we need to embed the a-priori available knowledge on the underlying system within a parameterised model set, 
under a prior distribution.  
The use of linearly parameterised model sets defined through orthonormal basis functions to represent partially unknown systems is a broadly used framework in system identification:  
it allows for the incorporation of prior knowledge, while maintaining the benefits (computational aspects) of linear parameterisations. 
Practically, it has been widely used for the modelling of physical systems, such as the thermal dynamics of buildings \cite{Virk1994,Reginato2009}.    
Note that although the goal of parameter exploration in formal verification has recently attracted quite some attention \cite{Batt2007,Frehse2008,Henzinger1996}, 
there are as of yet no general scalable results 
for the computation of the satisfaction function for nonlinearly-parameterised discrete-time LTI models. 
 Whilst in general linear time-invariant models with uncertain parameters
do not map onto a linearly-parameterised model set, 
we argue that a linearly-parameterised model set can encompass a relevant class of models. 
For instance, 
any asymptotically stable LTI model can be represented uniquely by its (infinite) impulse response \cite{Heuberger1995},   
and the coefficients of the impulse response define a linear parameterisation for this model. 
Further, 
the coefficients of the impulse response converge to zero, 
so that a truncated set of impulse coefficients can provide a good approximate model set with a finite-dimensional, linear parameterisation. 
This 
is only one possible instance of modelling by a finite set of orthonormal basis functions \cite[Chapters 4 and 7]{heuberger2005modelling},\cite{van1995system}, 
which can be selected to optimally incorporate prior knowledge:  
we conclude that, as an alternative to the use of a nonlinearly parameterised set of models, structural information (even when inexact) can be used to select a set of orthonormal basis functions, whose finite truncation defines a finite-dimensional linearly-parameterised model set indexed over the coefficients of the basis functions. 
Thus, in the following we consider a linearly parameterised model set $\Mset$ that encapsulates system $\s$, 
and specifically $\Mset=\{(A,B,C(\pa),D(\pa)),\pa\in\Theta\}$. 
 
 A system, or equivalently the mathematical model that represents it, 
satisfies a property if all the words generated by the model satisfy that property. 
Since properties are encoded over the external (input-output) behaviour of the 
system $\s$, 
which is the behaviour of $\M(\pa^0)$, $\pa^0\in\Theta$, 
we can equivalently assert that any property $\BLTL$ is verified by the system, $\s\vDash\BLTL$, if and only if it is verified by the unknown model representing the system, namely $\M(\pa^0)\vDash\BLTL$.
Introduce $\Theta_\BLTL$  to be the feasible set of parameters, 
such that for every parameter in that set the property $\BLTL$ holds, i.e., $\forall \pa \in \Theta_\BLTL:\M(\pa)\vDash\BLTL$. 
As such $ \Theta_\BLTL$ is characterised as the level set of the satisfaction function $f_\BLTL$, $\Theta_\BLTL=\{\pa\in \Theta:f_\BLTL(\pa)=1\}$. 

\subsection{ Safety Verification of Bounded-time Properties} \label{sec:LinSys}
 
Models $\M$ in the class $\Mset$ have the following representation $(A,B,C(\pa),0)$: 
\begin{align}%
\label{eq:model}
&\M(\pa):\quad \left\{\begin{array}{ll}
  \x{t+1}&=A\x{t}+Bu(t),\\
  \y{t,\pa}&=C(\pa)\x{t},
\end{array}\right.\end{align}
and are parameterised by $\theta\in \Theta \subset\Real^{pn}:$\(\pa=\operatorname{vec}(C) \)
with a prior probability distribution $\pd{\pa}$.
In addition to this strictly proper model class we will 
also allow for proper model $(A,B,C(\pa),D(\pa))$ where both the $C$ and  the $D$-matrices are parameterised and the parameterisation is $\pa=\operatorname{vec}([C \ D]))$.
For a given initial condition $\x{0}$ and input sequence, the output of the ``true" model $\y{t,\parTrue}$ is equal to the system output $\yo{t}$.
 
Given a measurement set-up as in Figure \ref{fig:System} with unknown parameter $\parTrue$. Then
$\ac{t}_{ex}$ and $\ym{t}_{ex}$ represent the input and the measured output, respectively, 
and $e(t)$ is an additive zero-mean, white, Gaussian-distributed measurement noise with covariance $\Sigma_e$ that is uncorrelated from the input. 
Furthermore $\ac{t}$ is assumed to be uncorrelated with the noise $e(t)$.
From this set-up 
$\Nsample$ samples are collected in a data set $Z^{\Nsample}=\{\ac{t}_{ex},\ym{t}_{ex}\}_{t=1}^\Nsample$. 

Therefore given the operating conditions of the experiment set-up the measured signal $\ym{t}_{ex}$ can be fully characterised: 
its probability density,  
conditional on the parameters $\pa$, is  
\begin{align*}%
&  \pd{Z^{N_s}|\pa} = \prod_{t=1}^{N_s} \pd{\ym{t}_{ex}|\pa}\\& = \frac{1}{ \sqrt{|\Sigma_e|^{N_s} (2\pi)^{pN_s}}} \exp\bigg[ \\& 
-\frac{1}{2} \sum_{t=1}^{N_s}(\y{t,\pa}-\ym{t}_{ex})^T \Sigma_e^{-1} (\y{t,\pa}-\ym{t}_{ex}) \bigg]
 \end{align*} 
and can be directly used in Proposition \ref{thm1}. 
This conditional density $ \pd{Z^{N_s}|\pa}$ depends implicitly on the given initial state $\x{0}_{ex}$ and, 
for the case of a given uncertainty distribution for $\x{0}_{ex}$, 
$ \pd{Z^{N_s}|\pa}$ should be marginalised as a latent variable \cite{Peterka1981a}. 
The a-posteriori uncertainty distribution is obtained as the analytical solution of the parametric inference in \eqref{BayesianID} \cite{Peterka1981a}.

Recall now that for a given specification $\BLTL$, 
we seek to determine a feasible set of parameters $\Theta_{\BLTL}$, 
such that the corresponding models admit property $\BLTL$, 
namely $\M(\pa)\vDash \BLTL, \ \forall \pa \in \Theta_{\BLTL}$.
Since models $\M(\pa)$ have a linearly-parameterised state space realisation as per \eqref{eq:model},  
it follows that when the set of initial states and inputs $\X_{ver}$ and $\mathbb U_{ver}$ are bounded polyhedra, 
the verification of a class of safety properties expressed by formulae with labels as in \eqref{eq:labelling} 
leads to a set of feasible parameters $\Theta_{\BLTL}$ that is a polyhedron, 
which can be easily computed.   
More precisely, 
the following theorem can be derived.

\begin{thm}[\cite{ACCSofie}]\label{Thm1}
Given a  bounded  polyhedral set (or equivalently a polytope) 
of initial states $\x{0}\in \X_{ver}$ and of inputs $u(t) \in \mathbb{U}_{ver}$ for $ t\geq 0$, 
and considering a labelling map as in \eqref{eq:labelling}, then
the feasible set $\Theta_\BLTL$ of the parameterised model set \eqref{eq:model} results in a polyhedron for properties $\psi$ composed of the LTL fragment  $\psi::=\letter|\nex\psi|\psi_1\wedge\psi_2$, with $\letter\in \alphabeth$.
\end{thm}

\begin{proof}[{of Theorem \ref{Thm1}}]
Let $\otimes$ denote the Kronecker product. 
Consider  the input set $\A_{ver}$ to be the convex hull of $U$, i.e. $\textmd{conv}(U)=\A_{ver}$. Similarly let the set of initial states be $\textmd{conv}(X_{ver})=\X_{ver}$. 
Let the model set be given as $\M(\pa)=(A,B,C(\pa),D)$.
We will temporarily assume that $D$ is set equal to zero. Afterwards we will show how to work with a parameterised  $D$.  Note that the syntax fragment $\psi::=\letter |\nex\psi|\psi_1\wedge\psi_2$ with $\letter\in \alphabeth=2^{AP}$ is equivalent to $\psi::=p |\nex\psi|\psi_1\wedge\psi_2$ with $p\in {AP}$.
 
\smallskip 
 
\textbf{1. } We claim that for every specification $\psi$ composed from the syntax fragment $\psi::=p |\nex\psi|\psi_1\wedge\psi_2$ and $\pa\in\Theta$, the words generated by a model $\M(\pa)=(A,B,C(\pa),0)$ with state $\x{t}$ satisfy the specification $\BLTL$, denoted $ <\M(\pa),\x{t}>\vDash \psi$, if and only if 
 \begin{equation}\label{eq:mafine} \left(\left( I_{n_\psi} \otimes\x{t} \right)^T\!\!N_\psi+K_\psi\right)  \pa  \leq B_\psi.\end{equation}
The matrices 
$N_\psi\in\mathbb{R}^{nn_\psi\times np},\allowbreak 
  K_\psi \in \mathbb{R}^{n_\psi\times np},\allowbreak  
  B_\psi \in \mathbb{R}^{n_\psi}$ in the above satisfaction relation 
have dimensions that are functions of the parametrisation and of the property dependent ``dimension'' $n_\psi$, 
and are obtained inductively over the syntax of the specification.  \\*
For any \emph{atomic propositions} 
 the model starting from state $\x{t}$ satisfies a property $\ap$, i.e.,  
 \( <\M(\pa),\x{t}>\vDash \ap \Leftrightarrow  A_{\ap} y\leq b_{\ap}\), 
 with $A_{\ap}\in \mathbb{R}^{1\times p}$ and $b_{\ap}\in \mathbb R$ we construct the matrices $N_{\ap}$, $K_{\ap}$ and $B_{\ap}$
as follows.
Consider $y(t)$ for a given $x(t)$ then \begin{align*}
A_\ap y(t)
&=A_\ap C(\pa) \x{t} 
=  \x{t}^T (  I_n\otimes A_\ap)\pa.
\end{align*} 
This yields
\(
N_{{\ap}}= (  I_n\otimes A_\ap)\,
\in\Real^{ n\times np},\;\allowbreak
 K_{{\ap}}=O_{1\times np}
\in\Real^{1\times np},\; \allowbreak
\allowbreak \mbox { and }  \allowbreak
  B_{{\ap}}= b_{\ap}
 \in\Real^{1\times 1}.\)
\\*
The \emph{next} operation  $\nex\psi_1$  with matrices ($N_{\psi_1}$,$K_{\psi_1}$,\allowbreak$D_{{\psi_1}}$,$b_{{\psi_1}}$) yields matrices
\begin{align*}
  N_{\nex {\psi_1}}&= \mathbf{1}_{|U|}\otimes\left(I_{n_{\psi_1}}\otimes A^T\right) N_{\psi_1} ,\\
  K_{\nex {\psi_1}}&=\mathcal U\left(I_{n_{\psi_1}} \otimes B \right)^TN_{\psi_1} + \mathbf{1}_{|U|} \otimes K_{\psi_1} ,\\ 
  B_{\nex {\psi_1}}&= \mathbf{1}_{|U|}\otimes B_{\psi_1} , 
 \shortintertext{where the $i$-th set of $n_{\psi_1}$ rows of $ \mathcal U\in\mathbb{R}^{|U|n_{\psi_1}\times m}$ is defined as }
&\left(I_{n_{\psi_1}}\otimes u_i^T\right)
\mbox{with $u_i\in U$}
\end{align*}
and where $n_{\nex\psi_1}=|U|n_{\psi_1}$.
This can be derived as
 \begin{align*} 
&<\M(\pa),\x{t}>\vDash  \nex \psi\Leftrightarrow\forall \ac{t}\in \mathbb{U}_{ver} :\\&\ \ \quad \left(\left(I_{n_{\psi_1}} \otimes\x{t+1}\right)^TN_{\psi_1}+K_{\psi_1}\right)\pa  \leq B_{\psi_1} ,\\ 
&\Leftrightarrow \forall \ac{t}\in \mathbb{U}_{ver}: \\& \left( \left( I_{n_{\psi_1}} \otimes A\x{t}\right)^T  N_{\psi_1} \right. \\
&\;\left.
\quad
+\left( I_{n_{\psi_1}} \otimes  Bu(t)\right)^TN_{\psi_1}+K_{\psi_1}\right) \pa \leq B_{\psi_1} .\end{align*}
Since
the above is an affine function in $u(t)$, the image of every $u(t)\in \textmd{conv}(U)=\mathbb{U}_{ver}$ can be expressed as a convex combination of the values at the vertices $u_i\in U$, c.f. \cite{belta2002control}.
Then an equivalent expression is 
\begin{align*} 
&\Leftrightarrow \forall u_i\in U: \Big( \left( I_{n_{\psi_1}} \otimes A\x{t}\right)^T  N_{\psi_1}   \qquad\qquad\\&  
+ \left(I_{n_{\psi_1}}\otimes u_i\right)^T\left(I_{n_{\psi_1}} \otimes B \right)^TN_{\psi_1}+  K_{\psi_1}\Big)\pa  \leq   B_{\psi_1} 
\shortintertext{which can be rewritten as}
&\Leftrightarrow  \Big(\mathbf{1}_{|U|}\otimes \left( I_{n_{\psi_1}} \otimes A\x{t}\right)^T  N_{\psi_1} 
+ \mathcal U\left(I_{n_{\psi_1}} \otimes B \right)^TN_{\psi_1}  \qquad\qquad\\& \hspace{1cm} + \mathbf{1}_{|U|} \otimes K_{\psi_1}  \Big)\pa  \leq \mathbf{1}_{|U|}\otimes  B_{\psi_1} .
\end{align*}
Having obtained $K_{\nex \psi}$, $D_{\nex \psi}$, and  $b_{ \nex \psi}$, now rewrite the first term to obtain $N_{\nex \psi}:$
\begin{flalign*}
  &\mathbf{1}_{|U|}\otimes\left(I_{n_{\psi_1}}\otimes \xT(t)\right)\left(I_{n_{\psi_1}}\otimes A^T\right) N_{\psi_1} \\
  &=\left(I_{|U|}\mathbf{1}_{|U|}\right)\otimes\left(I_{n_{\psi_1}}\otimes \xT(t)\right)\left(I_{n_{\psi_1}}\otimes A^T\right) N_{\psi_1} \\
&=\left(I_{|U|n_{\psi_1}} \otimes \xT(t)\right)\left(\mathbf{1}_{|U|}\otimes\left(I_{n_{\psi_1}}\otimes A^T\right) N_{\psi_1}\right).
\end{flalign*}
\\*
 The \emph{and} operation $\psi_1\wedge\psi_2$
for ($N_{\psi_1}$, $K_{\psi_1}$,$D_{{\psi_1}}$,$b_{{\psi_1}}$) and  ($N_{\psi_2}$, $K_{\psi_2}$,$D_{{\psi_2}}$,$b_{{\psi_2}}$)  with $n_{\psi_1\wedge \psi_2}=(n_{\psi_1}+n_{\psi_2})$ gives
\begin{align*}N_{\psi_1\wedge\psi_2}=\begin{bmatrix}
 N_{\psi_1}\\ N_{\psi_2}
\end{bmatrix}\!\!, 
  K_{\psi_1\wedge\psi_2}=\begin{bmatrix}
  K_{\psi_1}\\K_{\psi_2}
\end{bmatrix}\!\!,  
\;\,  B_{\psi_1\wedge\psi_2}=\begin{bmatrix}B_{\psi_1}\\B_{\psi_2}\end{bmatrix}\!\!.
\end{align*}
This can be derived from
   \begin{align*}
&<\M(\pa),\x{t}>\vDash\psi_1\wedge\psi_2\\
& \Leftrightarrow\bigwedge_{i\in \{1,2\}} \left(\left(I_{n_{\psi_i}} \otimes \x{t}\right)^TN_{\psi_i}+K_{\psi_i}  \right)\pa \leq B_{\psi_i}   \\
&\Leftrightarrow \!\! \left(\!\!\left(I_{n_{\psi_1\wedge \psi_2} }\otimes \x{t}\right)^T\!\!\begin{bmatrix}
  N_{\psi_1}\\ N_{\psi_2}
\end{bmatrix}\!\!+\!\begin{bmatrix}
  K_{\psi_1}\\K_{\psi_2}
\end{bmatrix}\right)  \pa \leq\!\! \begin{bmatrix}B_{\psi_1}\\B_{\psi_2}\end{bmatrix}
 .\end{align*}

\smallskip
 
\textbf{2. } 
The matrix-valued function 
\[\left(\left(I_{n_{\psi}} \otimes \x{0}\right)^TN_\psi+K_\psi \right) \pa \]
is affine in $ \xT(0)$ (for a fixed $\pa$), 
therefore its value at the initial condition $\x{0}\in\X_{ver}$ is a convex combination of the function values at the vertices $X_{ver}$ of $\X_{ver}$. 
Thus the satisfaction relation $ <\M(\pa), \x{0}>\vDash \psi$ represented by the multi-affine inequality holds uniformly over $\x{0}\in\X_{ver}$ if and only if it holds for the vertices of $\X_{ver}$. \\*
This gives a set of affine inequalities in $\pa$, 
thus the feasible set $\Theta_\psi$  is a polyhedron and is given as \[\left\{\pa\in\Theta:\bigwedge_{ \mathbf{x}_i\in X_{ver}} \!\!\!\!\!\left(\left(I_{n_{\psi}}\otimes \mathbf x_i\right)^TN_\psi+K_\psi \right)\pa   \leq B_\psi \right\}.\]
 The set $\Theta_\psi$ is a polyhedron, since it is formed by a finite set of half spaces.

\textbf{3.} To prove Theorem \ref{Thm1} we need to extend the results to  models with parameterised $D$.
The dynamics of model $(A,B,C,D)$ with both $C$ and $D$ fully parameterised can be reformulated as
\begin{align*}
\begin{bmatrix}\x{t+1}\\u(t+1)\end{bmatrix}&=\begin{bmatrix}A&B\\0&0\end{bmatrix}\begin{bmatrix}\x{t}\\u(t)\end{bmatrix} + \begin{bmatrix}0\\I\end{bmatrix}u(t+1)\\
y(t)&=\begin{bmatrix}C &D\end{bmatrix} \x{t}.
\end{align*}
Using the new matrices $(\tilde A,\tilde B, \tilde C(\pa),0) $ the obtained results still hold. For part \textbf{2.} set of vertices $X_{ver}$ needs to be extended with the vertices of $U$ as $X_{ver}\times U$. 
 \end{proof}

In the computation of the feasible set, the faces of the polyhedron $\Theta_\BLTL$ are shown to be a function of the 
vertices\footnote{A polytope can be written as the convex hull of a finite set of \emph{vertices}.}
 of the bounded set of initial states $\X_{ver}$ and of the set of inputs $\A_{ver}$, and are also expected to grow in number as a function of the time horizon of the property. \\ 
The result in Theorem \ref{Thm1} is valid for any finite composition of the LTL fragment $\psi::=\letter|\nex\psi|\psi_1\wedge\psi_2$, as such it only holds for finite horizon properties. Properties defined over the infinite horizon will be the objective of Section \ref{sec:ubp}. 
 
\subsection{Case Study: Bounded-Time Safety Verification
 }\label{ex:case}  
   
Consider a system $\s$ and verify whether the output $\yo{t}_{ver}$ remains within the interval $\mathcal I = \begin{bmatrix} 
  -0.5,\ 0.5\end{bmatrix}$, labeled as $\iota$, 
  for the next 5 time steps, under  $u(t)_{ver} \in\mathbb{U}_{ver}=[-0.2,\ 0.2]$ and $\x{0}_{ver}\in\{0_2\}=\X_{ver}$. 
  Introduce accordingly the alphabet $\alphabeth = \{\iota, \tau\}$ and the labelling map $L: L(y) = \iota, \forall y\in\mathcal I$, $L(y) = \tau, \forall y\in \Y \setminus \mathcal I$.
 Now check whether the following LTL property holds: 
$\s\vDash \bigwedge _{i=1}^{5} (\nex)^i \iota$.  
 We assume that system $\s$ can be represented as an element of a model set $\Mset$ with transfer functions characterised by second-order Laguerre-basis ones \cite{Heuberger1995} (a special case of orthonormal basis functions), which translates to the following parameterised state-space representation: 
\begin{align}%
 \label{eq:laguerre}
\!\!\!\! \!\!
\begin{array}{ll}
  \x{t+1}&\!\!\!\!  = \begin{bmatrix}
    a&0\\1-a^2&a
  \end{bmatrix} \x{t}+\begin{bmatrix}
   \cramped{ \sqrt{1-a^2}}\\(-a)\cramped{\sqrt{1-a^2}}
  \end{bmatrix}u(t), \\
  \y{t,\theta}&\!\!\!\! = \pa^T\x{t}\;. 
\end{array}
\end{align}
The parameter set is chosen as $\pa\in \Theta=[-10,10]^2$, whereas the coefficient $a$ is chosen to be equal to $0.4$.  
We select, as prior available knowledge on the system, 
a uniform distribution $\pd{\pa}$ on the model class, 
and pick a known variance $\var_e^2=0.5$ for the white additive noise on the measurement.   
The set of feasible parameters $\Theta_\BLTL\subset\Theta$ is represented in Figure \ref{fig:laguerre1} and is computed according to Theorem \ref{Thm1}. 
Based on the prior available knowledge, the confidence associated to $\pa_0\in \Theta_\BLTL$ amounts to $0.0165$\footnote{This is obtained by numerical computation of \eqref{eq:BayesianConf} with probability distribution $\pd{\pa}$. ntegrals are solved via the numerical integration tool in \texttt{Matlab}.}. 
In comparison to this value, after doing an experiment on the system with ``true parameter'' $\pa_0=[1\ 0]^T$ (Figure \ref{fig:laguerre1}) and with input signal $\ac{t}_{ex}$, a realisation of a white noise with a uniform distribution over $[-0.2,0.2]$, and measuring $\ym{t}_{ex}$ for 200 consecutive time instances 
the uncertainty distribution is refined as $\pd{\pa|Z^{N_s}}$. The resulting confidence  \eqref{eq:BayesianConf} in the property is increased to $0.779$. \\*
Along this line of experiments, 
we have repeated the test 100 times, 
for several instances of the parameter $\pa^0$ characterising the underlying system $\s$.  
In all instances, after obtaining 200 measurements the a-posteriori confidence represents the confidence in the safety of the system, 
as displayed in Table \ref{Table:Conf} via mean and variance terms.
 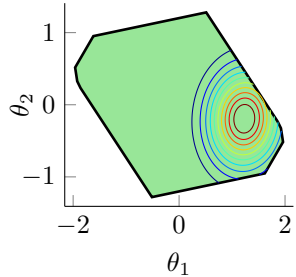
\begin{SCfigure}[1][h]
\begin{minipage}[b]{.4\columnwidth} 
  \centering \newlength\figureheight
    \newlength\figurewidth
    \setlength\figureheight{.8\linewidth}
    \setlength\figurewidth{.9\columnwidth}
   \hskip-1cm 
\definecolor{mycolor1}{rgb}{0.19608,0.80392,0.19608}%
\definecolor{mycolor2}{rgb}{0.00000,0.00000,0.56250}%
\definecolor{mycolor3}{rgb}{0.00000,0.43750,1.00000}%
\definecolor{mycolor4}{rgb}{0.00000,0.87500,1.00000}%
\definecolor{mycolor5}{rgb}{0.31250,1.00000,0.68750}%
\definecolor{mycolor6}{rgb}{0.75000,1.00000,0.25000}%
\definecolor{mycolor7}{rgb}{1.00000,0.81250,0.00000}%
\definecolor{mycolor8}{rgb}{0.93750,0.00000,0.00000}%
\begin{tikzpicture}

\begin{axis}[%
width=\figurewidth,
height=\figureheight,
colormap/jet,ylabel near ticks,
unbounded coords=jump,
view={0}{90},
scale only axis,
xmin=-2.15866290379835,
xmax=2.15866290379835,
xlabel={$\theta_1$},
ymin=-1.40852485670389,
ymax=1.40852485670389,
ylabel={$\theta_2$},
zmin=-1,
zmax=1,
axis x line*=bottom,
axis y line*=left,
axis z line*=left,
 y label style={at={(-0.1,0.5)}}
]

\addplot3[area legend,line width=1.0pt,fill=mycolor1,fill opacity=5.000000e-01,draw=black,forget plot]
 table {myfile-1.tsv};

\addplot[area legend,solid,draw=mycolor2,forget plot]
 table {myfile-2.tsv};

\addplot[area legend,solid,draw=blue,forget plot]
 table {myfile-3.tsv};

\addplot[area legend,solid,draw=blue,forget plot]
 table  {myfile-4.tsv};

\addplot[area legend,solid,draw=mycolor3,forget plot]
 table {myfile-5.tsv};

\addplot[area legend,solid,draw=mycolor4,forget plot]
 table{myfile-6.tsv};

\addplot[area legend,solid,draw=mycolor5,forget plot]
 table{myfile-7.tsv};

\addplot[area legend,solid,draw=mycolor6,forget plot]
 table{myfile-8.tsv};

\addplot[area legend,solid,draw=mycolor7,forget plot]
 table{myfile-9.tsv};

\addplot[area legend,solid,draw=red!25!orange,forget plot]
 table{myfile-10.tsv};

\addplot[area legend,solid,draw=mycolor8,forget plot]
 table{myfile-11.tsv};

\addplot[area legend,solid,draw=black!50!red,forget plot]
 table{myfile-12.tsv};

\end{axis}
\end{tikzpicture}%
    \end{minipage}
  \caption{Feasible set of parameters in $\Theta$, and contour lines of the quantity $\pd{\pa|Z^{N_s}}$, obtained for \(\pa^0=[1\ 0]^T\). }\label{fig:laguerre1}
\end{SCfigure} 
 \begin{table}[t]
\caption{Mean ($\mu$) and variance ($\var^2$) of the confidence obtained from  
      100 experiments with 200 measurements each. 
      }\label{Table:Conf}
      \centering
          \begin{tabular}[]{|p{1.3cm}p{0.9cm}p{0.82cm} |p{1.2cm}p{0.9cm}p{0.78cm}|}\hline &&&&&\\[-2ex]
    $\ \ \pa^0$&$\mu$&$\var^2$ &$\ \ \pa^0$& $\mu$&$\var^2$\vspace{0.1em}\\
    \hline\hline &&&&&\\[-2ex]
  $\begin{bmatrix}\mbox{-}1&\mbox{-}1
    \end{bmatrix}^T$&$0.348$& $0.073$ &$\begin{bmatrix}\ 1&\mbox{-}1
    \end{bmatrix}^T$& $ 0.491$&$ 0.085  $\\[0.7ex]
    $\begin{bmatrix}\mbox{-}1&\ 0
    \end{bmatrix}^T$&$0.705$& $ 0.060 $&$\begin{bmatrix}\ 1&\ 0
    \end{bmatrix}^T$&$  0.730 $&$ 0.056 $\\[0.7ex]
    $\begin{bmatrix}\mbox{-}1&\ 1
    \end{bmatrix}^T$&$0.492$&$0.086$& $\begin{bmatrix}\ 1&\ 1
    \end{bmatrix}^T$&$0.339$&$ 0.065$\\ \hline
  \end{tabular} 
\end{table}

\subsection{Verifying Unbounded-Time Properties  Using Invariant Sets}\label{sec:ubp}
In this section we extend the approach unfolded in Section \ref{sec:LinSys}, 
to hold on the LTL fragment $\psi::=\letter|\nex\psi|\psi_1\wedge\psi_2$ with additionally the {\it unbounded} invariance (safety) operator. 
Recall the form of the $k$-bounded and of the unbounded invariance operators, 
namely $\always^k\psi=\bigwedge_{i=0}^{k}\nex^i \psi$ and $\always\psi=\neg(\texttt{true}\until\neg \psi)$ respectively.  
The extension from a $k$-bounded operator, covered by the result in Theorem \ref{Thm1}, 
to the unbounded invariance one, 
is based on the concept of robust positive invariance \cite[Def. 4.3]{Blanchini2007}, 
recalled next.  
 \begin{definition} 
 For the system $\x{t+1}=A\x{t}+Bu(t)$, the set $\mathcal{S}\subseteq \X$ is said to be robustly positively invariant if, for all $\x{0}\in\mathcal{S}$ and $u(t)\in\A$, the condition $\x{t}\in \mathcal{S}$ holds for all $t\geq 0$. 
 \end{definition}

Recall that the feasible set $\Theta_\BLTL$ is defined as the set of parameters for which 
property $\BLTL$ holds, 
namely $\forall \pa \in \Theta_\BLTL:\M(\pa)\vDash\BLTL$.  
The satisfaction relation $\M(\pa)\vDash\BLTL$ depends implicitly on the set of initial states $\x{0}\in\X_{ver}$ and on the set of inputs $\A_{ver}$. 
Let us extend the definition of the feasible set to explicitly account for its dependence on the set of initial conditions:  
given a bounded and convex set $\mathcal{S}\subset \X$, let $\Theta_{\psi}(\mathcal{S})$ be defined as  the set of parameters in $\Theta$ for which the parameterised models  $\M(\pa)$ initialised with $\x{0}\in  \mathcal{S} $  satisfy $\psi$ over input signals $u(t)\in \A_{ver}$ $t\geq 0$. Hence the feasible set $\Theta_\psi$ can be written as a function of the set of initial states $\X_{ver}$, that is $\Theta_\psi\left(\X_{ver}\right)$. Thus
the extended map $\Theta_\psi\left(\cdot\right)$ takes subsets of the state space into subsets of the parameter space.  
 Note that if $\mathcal{S}$ is a robustly positively invariant set that includes the set of initial states $\X_{ver}\subseteq\mathcal{S}$, then for all $\pa\in\Theta_{\BLTL}(\mathcal{S})$ the models $\M(\pa)$ satisfy $\BLTL$ over all infinite-time model traces $\x{t}$: this allows to state that $\M(\pa)\vDash \always \psi$. 
{We can show that the following holds.}
 
\begin{lem}\label{lem:decreasing}
The function $\Theta_\BLTL(\cdot):2^\X\rightarrow 2^\Theta$, 
for specifications obtained as  $\psi::= \letter \mid\nex\psi\mid \psi_1\wedge\psi_2$, is monotonically decreasing: 
that is if $\mathcal{S}_1\subseteq \mathcal{S}_2$, 
then $\Theta_\BLTL(\mathcal{S}_2)\subseteq\Theta_\BLTL(\mathcal{S}_1).$ 
\end{lem} 
  \begin{proof} 
We leverage the notation used in the proof of Theorem \ref{thm1}. Provided that the parameterised model is given as $(A,B,C(\pa),0)$, 
we show that any  $\pa\in\Theta_\BLTL(\mathcal{S}_2)$ is also an element of $\pa\in\Theta_\BLTL(\mathcal{S}_1)$.
Suppose $\mathcal{S}_2$ has a finite number of vertices $\mathbf{x}_i\in\mathcal{V}\left( \mathcal{S}_2\right) $, then  for any $\pa\in  \Theta_\BLTL(\mathcal{S}_2)$:
  \begin{align*}
  \textstyle \bigwedge_{ \mathbf{x}_i\in \mathcal{V}\left( \mathcal{S}_2\right) } \left((I_{n_{\psi}}\otimes   \mathbf{x}_i)^TN_\psi+K_\psi\right)\pa  \leq B_\psi \end{align*}
and for every  $\mathbf{x}\in \mathcal{S}_2$   
\begin{equation*}  \left((I_{n_{\psi}}\otimes   \mathbf{x})^TN_\psi+K_\psi\right)\pa  \leq B_\psi.\end{equation*}  
Since the vertices $\mathbf{x}_j\in\mathcal{V}\left(\mathcal{S}_1\right)$ are also elements of $\mathcal{S}_2$, then 
  \begin{align*}\textstyle\bigwedge_{ \mathbf{x}_j\in \mathcal{V}\left( \mathcal{S}_1\right) }  \left((I_{n_{\psi}}\otimes    \mathbf{x}_j)^TN_\psi+K_\psi\right)\pa  \leq B_\psi \end{align*}
and $\pa\in \Theta_\BLTL(\mathcal{S}_1)$. This reasoning can be trivially extended to include parameterised $D$ matrices.
Increasing the number of vertices of $\mathcal{S}_1$ and $\mathcal{S}_2$, does not change the result, hence the same holds if $\mathcal{S}_1$ and $\mathcal{S}_2$ are convex sets. \mbox{ }
 \end{proof}

Based on the result in 
Lemma \ref{lem:decreasing}, 
we conclude that the maximal feasible set $\Theta_{\always\psi}$ is obtained as a mapping from the minimal robustly positively invariant set $\mathcal{S}$ that includes $\X_{ver}$:  
$\Theta_{\always\psi}=\Theta_{\psi}(\mathcal{S})$. 
This leads next to consider under which conditions such minimal robustly positively invariant set $\mathcal{S}$ can be exactly computed or approximated.   

\subsubsection*{Feasible set for invariance properties with $\X_{ver}=\{0_n\}$}
For $\X_{ver}=\{0_n\}$, 
assuming a bounded interval $\A_{ver}$ with the origin in its interior, 
and under some basic assumptions on the dynamics (to be shortly discussed), 
the minimal robustly positively invariant set can be shown to be a bounded and convex set that includes the origin \cite{Blanchini2007}. 
Maintaining the condition of $\A_{ver}$ being bounded and having the origin in its interior,  
we first consider the case that $\X_{ver}=\{0_n\}$ and characterise $\mathcal{S}$ via tools available from set theory in systems and control; 
thereafter we look at extensions to more general sets of initial states $\X_{ver}$.  

Assume that $\A_{ver}$ includes the origin,  and denote the forward reachability mappings initialised with \(\mathcal{R}^{(0)}:=\{0_n\}\subset \X\) as  
 \begin{align}%
 \label{eq:ForwardReachability}
\mathcal{R}^{(i)}&:=  \operatorname{Post}(\mathcal{R}^{(i-1)} ), 
 \end{align}
with set operation $\operatorname{Post}(X):=\{\mathbf{x}'=A\mathbf{x}+Bu, \mathbf{x}\in X, u\in \A\}$.
Denote the limit reachable set 
as $\mathcal{R}^\infty =\cramped{\lim_{i\rightarrow \infty}\mathcal{R}^{(i)}}$.   
From literature we recall that properties of these $i$-step reachable sets, as given in 
 \cite{Blanchini2007}  include the following: 
for a reachable pair $(A,B)$ and an asymptotically stable matrix $A$, 
the $\infty$-reachable set $\mathcal{R}^{\infty}$ is bounded and convex \cite[Proposition 6.9]{Blanchini2007}. 
The $k$-step reachable set converges to the $\infty$-reachable set via \eqref{eq:ForwardReachability}, since it is monotonically increasing $\cramped{\mathcal{R}^{(i)}\subseteq\mathcal{R}^{(i+1)}}$. 
Moreover, $\mathcal{R}^\infty$ is the minimal robustly positively invariant set for the system, 
so that any 
positively invariant set includes $\mathcal{R}^\infty$ \cite[Proposition 6.13]{Blanchini2007}.  
Thus, starting from $\x{0}=0_n$,  all $\x{t}\in\mathcal{R}^{\infty}$, 
and furthermore of interest to this work {we conclude that} $\Theta_{\always^k\psi}\!\!=\cramped{\Theta_{\psi}\big(\mathcal{R}^{(k)}\big)}$ and $\Theta_{\always\psi}=\cramped{\Theta_{\psi}\big(\mathcal{R}^{\infty}\big)}$.

\subsubsection*{Feasible set for invariance properties under polytopic sets of initial states} 

More generally, if $\X_{ver}\subseteq \mathcal{R}^{\infty}$ and ceteris paribus, 
then $\mathcal{R}^{\infty}$ is the minimal robustly positively invariant set that includes $\X_{ver}$, 
and $\Theta_\BLTL(\mathcal{R}^\infty)=\Theta_{\always\BLTL}$.  
For finite iterations the reachable sets $\mathcal{R}^{(i)}$ are polytopes, 
and if $\mathcal{R}^{(i)}=\mathcal{R}^{(i+1)}$, 
then $\mathcal{R}^{(i)}=\mathcal{R}^{\infty}$\!. 
Though the iterations can stop in finite time, 
in general the number of iterations to obtain $\mathcal{R}^{\infty} $ can be infinite. 
Whilst the minimal robustly positively invariant set is not necessarily closed or a polytope,  
there exist methods to approximate $\mathcal{R}^\infty$ as detailed in \cite{Blanchini2007}.   
For instance, for stable systems, $\mathcal{R}^{(k)}$ is shown to converge to $\mathcal{R}^\infty$, 
in the sense that for all $\epsilon>0$ there exists $\bar{k}$ such that for $k\geq \bar{k}$,
$\mathcal{R}^{(k)}\!\!\subseteq\mathcal{R}^\infty\!\!\subseteq(1+\epsilon)\mathcal{R}^{(k)}$ \cite[Proposition 6.9]{Blanchini2007}.

Recall that the maximal feasible set $\Theta_{\always\psi}$ is obtained as a mapping from the minimal robustly positively invariant set $\mathcal{S}$ including $\X_{ver}$, 
that is $\Theta_{\always\psi}=\Theta_{\psi}(\mathcal{S})$. 
Let us extend the study 
to the case where the conditions $\X_{ver}=\{0_n\}$ or its extension $\X_{ver}\subseteq \mathcal{R}^\infty$ do not apply, 
while the condition on the bounded set $\A_{ver}$ is maintained, that is $0\in \A_{ver}$. 
Consider the more general case where the set of initial states is a polytope but not necessarily a subset of  $\mathcal{R}^\infty$. 
Denote the union of the forward reachability mappings initialised with \(\mathcal{R}^{(0)}_{\X_{ver}}:=\X_{ver}\subseteq \X\) as  
 \begin{align}\label{eq:ForwardReachability2} 
\mathcal{R}^{(i)}_{\X_{ver}}&:=  \mathcal{R}^{(i-1)}_{\X_{ver}}\cup\operatorname{Post}(\mathcal{R}^{(i-1)}_{\X_{ver}} )\ .
 \end{align} 
This set is also known in the literature as the {\it reach tube}. 
The corresponding set for infinite time is denoted as $\cramped{\mathcal{R}^\infty_{\X_{ver}} =\lim_{i\rightarrow \infty}\mathcal{R}^{(i)}_{\X_{ver}}}$.  
Notice that if $\X_{ver}\subseteq \mathcal{R}^{\infty}$, 
then $\cramped{\mathcal{R}^{\infty}=\mathcal{R}^{\infty}_{\X_{ver}}}$. 
The iteration is monotonically increasing $\cramped{\mathcal{R}^{(i)}_{\X_{ver}}\subseteq \mathcal{R}^{(i+1)}_{\X_{ver}}}$,  
and whenever $\mathcal{R}^{(i)}_{\X_{ver}} = \mathcal{R}^{(i+1)}_{\X_{ver}}$ it stops after a finite number of iterations with $ \mathcal{R}^{\infty}_{\X_{ver}}=\mathcal{R}^{(i)}_{\X_{ver}}$. 
Of course, also in this more general case, the number of iterations can be unbounded, 
however the convergence properties of $\mathcal{R}^{(i)}$ extend seamlessly to the case of sets $\mathcal{R}^{(i)}_{\X_{ver}}$. 
Since $ \mathcal{R}^{(i)}_{\X_{ver}}$ is a union of polytopes, it is not guaranteed to be a convex set. 
Still, it can be shown via the proof of Theorem \ref{Thm1} that the computation of the feasible set $\Theta_\BLTL(\mathcal{S})$ boils down to that of $\Theta_\BLTL\big(\operatorname{conv} (\mathcal{S})\big)$.  

\begin{rem}
Let us illustrate the convergence property for sets $\mathcal{R}^{(i)}_{\X_{ver}}$ as follows.  
For every vertex $\mathbf{x}^{i}(0) \in \X_{ver}$, 
select a decomposition $\mathbf{x}_r^i+\mathbf{x}_s^i$ with $\mathbf{x}_r^i\in\mathcal{R}^\infty$, 
which minimises \(\|\mathbf{x}_s^i\|\) for a chosen vector norm $\| \cdot\|$. 
Since every element $\x{0}\in \X_{ver}$ is a convex combination of the vertices $\mathbf{x}^{i}(0)$, it follows that for all $\x{0}\in \X_{ver}$:  
\begin{align*}
\x{0}&=\sum_{i} a_i \mathbf{x}^i(0)   = \sum_{i} a_i \mathbf{x}_r^i(0)+ \sum_{i} a_i   \mathbf{x}_s^i(0)\\
&\in \operatorname{conv}(\mathbf{x}_r^i(0))+ \operatorname{conv}(\mathbf{x}_s^i(0) )\subseteq \mathcal{R}^\infty+ \bar{\X}_{ver}, 
 \end{align*} 
with $\cramped{\sum_i a_i=1}$ for $a_i
 \geq0$ and where $\bar{\X}_{ver}=\operatorname{conv}(\mathbf{x}_s^i(0) )$. 
 We obtain that 
\(\cramped{\X_{ver}\subseteq \mathcal{R}^\infty + \bar{\X}_{ver}} \), 
and that the minimal positively invariant set $\mathcal{R}^\infty_{\X_{ver}}$ can be bounded by 
\(\cramped{\mathcal{R}^\infty + 
\lim_{k\rightarrow \infty} \bigcup_{i=0}^k A^i \bar{\X}_{ver}.}
\) 
Under condition of asymptotic stability on $A$, 
necessary for $\mathcal{R}^{\infty}$ to be a bounded and convex polytope, $A^i \bar{\X}_{ver}$\! will converge to $\{0_n\}$. 
Thus, the iteration $\mathcal{R}^{(k)}_{\X_{ver}}$ is monotonically increasing and bounded, hence it converges. 
If $\bar{\X}_{ver}$ includes the origin in its interior then there exists a finite iteration such that $ \cramped{\bigcup_{i=0}^k A^i \bar{\X}_{ver}=\bigcup_{i=0}^{k+1} A^{i} \bar{\X}_{ver}}$. 
Moreover, for any reachable pair $(A,B)$ and asymptotically stable $A$, 
the closure of the minimal robustly positively invariant set $\mathcal{R}^{\infty}_{\X_{ver}}$ includes the origin.  \end{rem}

\subsubsection*{Robust approximations of the feasible set via $\Theta_\psi(\cdot)$} 
In order to exploit convergence in the computation of the feasible set for invariance properties, 
we need to bound the error incurred with the use of approximations of the sets $\mathcal{R}^\infty_{\X_{ver}}$ or $\mathcal{R}^\infty$. 
Let $\mathcal{B}$ denote a unit ball centred at the origin and let the Hausdorff distance between sets $\mathcal{R}_{1}$ and $\mathcal{R}_{2}$ be defined as
\begin{align*}%
\delta_H(\mathcal{R}_1,\mathcal{R}_2)=\inf\{\epsilon\geq 0|\mathcal{R}_1\subseteq \mathcal{R}_2+\epsilon \mathcal{B}, \mathcal{R}_2 \subseteq\mathcal{R}_1+\epsilon \mathcal{B}\}. 
\end{align*}
 
{We can show that the following holds.} 
\begin{lem}\label{lem:robust}
Consider a polytope $\mathcal{R}$, and a property $\psi$ comprised of $\psi::= \letter|\nex \psi| \psi_1\wedge \psi_2$, with $\letter\in \alphabeth$,  
for which $\Theta_{\psi}(\mathcal{R})$ is a non-empty polytope with vertices $v_i$ and the origin in its interior. 
Let $A$ be bounded as $\|A\|_2\leq1$. 
Then for any $\epsilon_x\geq 0$, 
 \begin{align}\label{eq:n12}
 &\Theta_{\psi}(\mathcal{R}+\epsilon_x \mathcal{B}) \subseteq \Theta_{\psi}(\mathcal{R}) \subseteq \Theta_{\psi}(\mathcal{R} +\epsilon_x \mathcal{B})+\epsilon_\pa \mathcal{B}\\
 &\textrm{if }\epsilon_\pa\geq   \frac{\epsilon_x \epsilon_{p} \max_i(\|v_i\|)^2}{1+\epsilon_x\epsilon_p\max_i( \|v_i\|) },\textrm{ for } \epsilon_p:=   \max_{p\in AP}\frac{\|A_{p}\|_2}{|b_{p}|  }.\notag\qquad\end{align}
 \end{lem}
\begin{proof}   
\noindent\textbf{1. \( \Theta_{\psi}(\mathcal{R}+\epsilon_x \mathcal{B}) \subseteq \Theta_{\psi}(\mathcal{R})\)} \\*
Based on the definition of this set (c.f. the proof of Theorem \ref{Thm1}), the set operation $\Theta_{\psi}(\cdot)$ is monotonically decreasing. Therefore 
\( \Theta_{\psi}(\mathcal{R}+\epsilon_x \mathcal{B}) \subseteq \Theta_{\psi}(\mathcal{R})\) holds.

\noindent\textbf{2. \( \Theta_{\psi}(\mathcal{R}) \subseteq \Theta_{\psi}(\mathcal{R}+\epsilon_x\mathcal{B})+\epsilon_\pa \mathcal{B}\)} \\*
Consider the case where the model is $(A,B,C(\pa),0)$.
To prove \eqref{eq:n12}, we first find a $\epsilon_\pa$ as a function of $\epsilon_x$ such that \begin{equation}\label{eq:n13}\Theta_{\psi}(\mathcal{R}) \subseteq \Theta_{\psi}(\mathcal{R}+\epsilon_x\mathcal{B})+\epsilon_\pa \mathcal{B}.\end{equation}
Let $v_i$ be the vertices of the polytope $v_i\in\mathcal{V}\left(\Theta_{\psi}(\mathcal{R})\right) $, then \eqref{eq:n13} holds if and only if $v_i\in \Theta_{\psi}(\mathcal{R}+\epsilon_x\mathcal{B})+\epsilon_\pa \mathcal{B} $. Equivalently, this means that there exists a $r_\pa\in \epsilon_\pa \mathcal{B} $ such that $v_i-r_\pa\in \Theta_{\psi}(\mathcal{R}+\epsilon_x\mathcal{B})$. 
This is equivalent to demanding that for every $\mathbf{x}_j^T\in \mathcal{V}\left(\mathcal{R}\right)$, $v_i\in\mathcal{V}\left(\Theta_{\psi}(\mathcal{R})\right) $ and  $r_x\in \epsilon_x\mathcal{B}$, there exists a vector $r_\pa\in \epsilon_\pa \mathcal{B}$:
 \begin{align}  \left((I_{n_{\psi}}\otimes (\mathbf{x}_j^T+r_x^T))N_\psi+K_\psi\right)  \left(v_i -r_\pa\right) & \leq B_\psi\notag \\\Leftrightarrow \quad \left((I_{n_{\psi}}\otimes \mathbf{x}_j^T)N_\psi+K_\psi\right)\left(v_i -r_\pa\right) \notag\\
+  \left((I_{n_{\psi}}\otimes r_x^T)N_\psi\right)  \left(v_i -r_\pa\right)&\leq B_\psi.\notag\end{align}
Take $(v_i-r_\pa)=(1-\alpha_i) v_i$ with $\alpha_i\in [0,1)$, then 
 \begin{align} 
 \left((I_{n_{\psi}}\otimes \mathbf{x}_j^T)N_\psi+K_\psi\right)  (1-\alpha_i) v_i&\notag\\
+ \left((I_{n_{\psi}}\otimes r_x^T)N_\psi\right)   (1-\alpha_i) v_i&\leq B_\psi\notag\\
\Leftrightarrow \quad(1-\alpha_i)   (I_{n_{\psi}}\otimes r_x^T)N_\psi    v_i&\leq\alpha_i B_\psi. \label{eq:n17}
\end{align}
Separate the matrix $N_\psi$ and $B_\psi$ into its block matrices $N_\psi^j=[N_\psi]_{\{1+(j-1)n :nj\}\times\{ 1:n\}}$ and $ B^j=[B^j_\psi]_{j}$ such that inequality \eqref{eq:n17} is equivalent to the set of inequalities
    \begin{align}
(1-\alpha_i)  r_x^TN^j_\psi  v'_i&\leq\alpha_i b^j, \mbox{ for }  j=1,\ldots ,n_\psi   \\
 \Leftrightarrow\quad  r_x^TN^j_\psi  v'_i&\leq\frac{\alpha_i}{(1-\alpha_i)} b^j\ \ .&
\end{align}
Given that $0\in\Theta_\psi(\mathcal{R})$,  it follows that $b_j\geq 0$ for $j=1,\ldots,n_\psi$
\begin{align*}
 \max_{j}  \left( r_x^TN^j_\psi  v'_i\right)  (b^j)^{-1}\leq\frac{\alpha_i}{(1-\alpha_i)} \ \ .&
\end{align*}
The term on the left can be upper bounded based on the Cauchy-Schwarz inequality  
\begin{align*}
 &\max_{j}  \left( r_x^TN^j_\psi  v'_i\right)  (b^j)^{-1}\leq \max_{j}  \|(N^j_\psi)^T r_x \|_2 \|v'_i\|_2 (b^j)^{-1}\\
 &\quad\leq \max_{j}  \|(N^j_\psi)^T\|_2\| r_x \|_2 \|v'_i\|_2 (b^j)^{-1}  \mbox{ and } \|r_x\|_2\leq \epsilon_x\\
 &\quad\leq \epsilon_x \epsilon_p\|v'_i\|_2.
\end{align*}
The last inequality follows from the introduction of  the precision of the labelling, 
denoted as $\epsilon_p$, and defined as \begin{align}
\epsilon_p&= \max_{p\in AP}\frac{\|A_{p}\|_2}{|b_{p}| }. 
\end{align}
Remember that $\|L\otimes K\|_2=\|L\|_2\|K\|_2$. Then based on Theorem \ref{Thm1} and on the condition $\|A\|_2\leq 1$, 
it can be shown that 
\[ \max_{j}\|(N^j_\psi)^T \|_2 |b^j|^{-1} \leq \max_{p\in AP}\frac{\|A_{p}\|_2}{|b_{p}| }. \]
Note that $\frac{\alpha_i}{(1-\alpha_i)} $ monotonically increases with $\alpha_i$ for $\alpha_i\in [0,1)$. 
Therefore a bound on $\alpha_i$ can be found as
\begin{align}
\alpha_i  = (\epsilon_x  \epsilon_p\|v_i\|) /(1+\epsilon_x \epsilon_p\|v_i\|)  \mbox{ for }  j=1,\ldots ,n_\psi. 
\end{align}
It follows that \eqref{eq:n13} holds if 
\begin{align}
\epsilon_\pa&= \max(\|v_i\|_2) \frac{\epsilon_x \epsilon_p \max(\|v_i\|_2)}{1+\epsilon_x\epsilon_p \max(\|v_i\|_2)}. 
\end{align}
For the case that the model is parameterised in both $S$ and $D$, i.e., $(A,B,C(\pa),D(\pa))$ the derivation is a bit more cumbersome but can be repeated with no change to the end result.\end{proof}  

\medskip

Let us briefly discuss the conditions under which Lemma \ref{lem:robust} is applicable. 
The condition that $\Theta_\psi(\mathcal{R})$ is not empty is raised to avoid the trivial case where $\Theta_\psi(\mathcal{R})=\emptyset$ \eqref{eq:n12} holds for all $\epsilon_\pa$. 
The condition that $\Theta_\psi(\mathcal{R})$ is a polytope and hence bounded is necessary to obtain a bounded Hausdorff distance. 
This distance quantifies the difference between two sets, and is a necessary step to bound the approximation error. 
The requirement that $\Theta_\psi(\mathcal{R})$ includes the origin is a sufficient condition and relates to well-posedness for bounded input sets including the origin. 
When considering invariance properties defined for $\cramped{0\in \A_{ver}}$ and for any polytope $\X_{ver}$, 
the requirement that $\cramped{0_n\in\Theta_\psi(\cdot)}$ is necessary for $\Theta_{\always \psi}$ to be non-empty: 
this can be intuitively illustrated by noting that under an assumption of asymptotic stability for $A$, 
for any $\pa$ and for $u(\cdot)=0$ the output $\y{t,\pa}$ of the model in \eqref{eq:model} converges to 0.  
Hence for a property to be satisfied under these conditions it should at least hold for the zero output, which is equivalent to demanding that it holds for $\pa=0_n$.  
For any atomic proposition $\ap\in AP$ (see Equation \eqref{eq:labelling}) it can be shown that there is an invertible mapping between the row vectors, proportional to the normals of the faces of the polyhedral set $\Theta_{\ap}(\x{0})$, and the initial state $\x{0}$. 
Therefore, if $\mathcal{R}^{(k)}$ has the origin in its interior, then $\Theta_{\ap}(\mathcal{R}^{(k)})$ has to be bounded, and as a consequence so has any feasible set comprising this atomic proposition.  
This holds for $k\geq n$ if $(A,B)$ is a reachable pair and if $\A_{ver}$ has $0$ in its interior. Under the same conditions there exists a $k$ such that $\mathcal{R}^{(k)}_{\X_{ver}}$ has $0_n$ in its interior.
The generalisation to the case dealing with an Hausdorff distance of the feasible set for invariance properties with a set of inputs $0\not\in \A_{ver}$ is outside of the scope of this work.

\subsubsection*{Convergence properties} 
We can employ Lemma \ref{lem:robust} to bound the Hausdorff distance between  $\Theta_\psi(\mathcal{R}^{(k)}_{\X_{ver}})$ and  $\Theta_{\always\psi}$.  
If $\X_{ver}=\{0_n\}$ and the spectral radius of $A$ is strictly less than $1$ (that is $\rho(A)<1$), then the Hausdorff distance can be bounded as
    \begin{align}%
    \label{HausForward}
    \delta_H(\mathcal{R}^{(k)},\mathcal{R}^{\infty})\leq \epsilon(k) :=  \|A^k\|_{2}\max_{u\in \A}{(|u|)}  c_1,
    \end{align}
    with $c_1$ a bound on $\sum_{i=0}^{\infty} \|A^i B\|$, which is the peak-to-peak performance of the dynamical system formed by $(A,B)$.
In case that $\X_{ver}\not \subseteq \mathcal{R}^\infty $ then  the forward reachable  iteration can be rewritten as
\[\cramped{\mathcal{R}^{(k)}_{\X_{ver}}= \big( \bigcup_{i=0}^k A^i {\X}_{ver}\big)
+\mathcal{R}^{(k)}.}
\]
The Hausdorff norm can be bounded as
\[\cramped{\delta_H(\mathcal{R}_{\X_{ver}}^{(k)}, \mathcal{R}_{\X_{ver}}^{\infty})}\leq \cramped{\epsilon(k)+\|A^{k+1}\|_2  \delta_H\left(\X_{ver},\{0_n\}\right)}.\]
Note that for $\rho(A)<1$ the norm $\|A^{k}\|_2\rightarrow 0$ for $k\rightarrow \infty$.
In case the conditions of Lemma \ref{lem:robust}  on $\cramped{\mathcal{R}_{\X_{ver}}^{(k)}\subseteq \X}$ and $\Theta_\psi\big(\mathcal{R}_{\X_{ver}}^{(k)}\big)$ hold,  
the Hausdorff distance $\delta_H(\Theta_{\always^k\psi},\Theta_{\always\psi})$ can be bounded by 
\begin{equation}\label{Haus_feas}  
\|A^k\|_{2}\max_i(\|v_i\|)^2 \epsilon_{p} \big(\max_{u\in \A}{(|u|)}  c_1 +\|A\| \delta_H(\X_{ver},\{0_n\})\big).
\end{equation}  
\subsubsection*{Use in the verification of unbounded-time properties} 
Based on the convergence properties of the feasible set, 
the asymptotic behaviour of the confidence computed in Proposition \ref{thm1} can be stated.   
\begin{cor}[Convergence]\label{thm:cvg}
Under the conditions of Lemma \ref{lem:robust}, 
for a Gaussian distribution $\pd{\pa}\sim \mathcal{N}\left(\mu_\pa,\Var_\pa\right)$ with a covariance $\Var_\pa\succ0$, 
$\p{\pa\in \Theta_{\always^k\psi}}\rightarrow \p{\pa\in \Theta_{\always\psi}}$ for $k\rightarrow \infty$. 
\end{cor}
\begin{proof}[of Corollary \ref{thm:cvg}]
 For a strictly positive $\Var_\pa$, 
 the Gaussian density distribution takes finite values over the parameter space, 
 therefore the convergence of a monotonically-decreasing polytope over the parameter space induces the convergence of the associated probability measure. 
 \end{proof}

Theorem \ref{Thm1} can now be generalised to include unbounded-time invariance properties as follows.
\begin{thm}\label{Thm2}
Consider a polytopic set of initial states $x(0)\in \X_{ver}$, 
inputs $u(t) \in \mathbb{U}_{ver}$ for $ t\geq 0$, 
and a labelling map as in \eqref{eq:labelling}. 
Let $\hat{\mathcal{R}}^{\infty}_{\X_{ver}}$ be a polytopic superset of the minimal robustly positively invariant set that includes $\X_{ver}$, 
denoted as ${\mathcal{R}}^{\infty}_{\X_{ver}}$;  
then the feasible set admits a polyhedral subset $\hat{\Theta}_{\psi} \subset \Theta_{\psi}$ for every specification $\psi$ expressed within the LTL fragment $\psi:=\letter|\nex\psi|\psi_1\wedge\psi_2|\always \psi$, 
and if $\hat{\mathcal{R}}^{\infty}_{\X_{ver}}={\mathcal{R}}^{\infty}_{\X_{ver}}$\ then  $\hat{\Theta}_{\psi}=\Theta_{\psi}$.\end{thm} 
\begin{proof}
Every property $\phi::=p|\nex\psi|\psi_1\wedge\psi_2|\always \psi $ with $p\in AP$ can be rewritten as $\always \psi_1\wedge \psi_2$ where $\psi_1$ and $\psi_2$ have syntax $\psi: :=p|\nex\psi|\psi_1\wedge\psi_2 $.

For the set of initial states $\X_{ver}$, a property $\psi$ is invariant 
 \begin{align*}
\left<\M(\pa),\x{0}\right>\vDash \always \psi, \,  \forall  \x{0}\in\X_{ver}
\end{align*} 
 if and only if
   $\forall x\in \mathcal{R}^{\infty}_{\X_{ver}}:\left<\M(\pa),x\right>\vDash \psi $.
Let $\hat{\mathcal{R}}^{\infty}_{\X_{ver}}$ be a polytopic superset of ${\mathcal{R}}^{\infty}_{\X_{ver}}$ with a finite set of vertices $v_{\mathcal{R}}\in V_{\mathcal{R}}$, 
then the subset approximation of the feasible set $\Theta_{\always \psi}$ follows as 
$\Theta_{\always \psi}\supseteq\hat{\Theta}_{\always \psi}=$
 \begin{align*}\left\{\pa\in\Theta:\bigwedge_{v_{\mathcal{R}}\in V_{\mathcal{R}}}  \left( (I_{n_{\psi} }\otimes v_{\mathcal{R}}^T)N_\psi+K_\psi \right) \pa  \leq b_\psi\right\}
 \end{align*}
 where $\hat{\Theta}_{\always \psi}\subseteq {\Theta}_{\always \psi}$.
 Note that if $\hat{\mathcal{R}}^{\infty}_{\X_{ver}}={\mathcal{R}}^{\infty}_{\X_{ver}}$ then $\hat{\Theta}_{\always \psi}= {\Theta}_{\always \psi}$.
 The feasible set of $\always \psi_1\wedge \psi_2$ is equal to $\Theta_{\always \psi_1\wedge \psi_2}=\Theta_{\always \psi_1} \cap\Theta_{\psi_2}$. And $ \Theta_{\always \psi_1\wedge \psi_2}$ can be upper and lower bounded as 
 $\hat{\Theta}_{\always \psi_1} \cap\Theta_{\psi_2}\subseteq \Theta_{\always \psi_1\wedge \psi_2}\subseteq \Theta_{\always^k \psi_1} \cap\Theta_{\psi_2}$ with $k\in \mathbb{N}$.
 This proves Theorem \ref{Thm2} for the case where the model is $(A,B,C(\pa), 0)$. The additional parameterisation of $D$ does not change the reasoning.
 \mbox{ }
\end{proof} 

The extension beyond the LTL fragment discussed above may lead to feasible sets that are in general not convex and are therefore beyond the scope of this work. 

 \subsection{Case Study (cont.): Unbounded-Time Safety Verification}\label{ex:casecont}

We study convergence properties for the safety specification $\iota$ considered in the case study in Section \ref{ex:case} maintaining the same operating conditions as before for the safety verification and the experiment. 
In Figure \ref{freach} the forward reachability sets $\mathcal{R}^{(k)}$ with $k=1,\ldots,20$ are obtained for the model dynamics in \eqref{eq:laguerre}. 
Figure \ref{littleplots} (upper plot) displays bounds $\epsilon(k)$ on the Hausdorff distances $\delta_H(\mathcal{R}^{(k)},\mathcal{R}^{\infty})$ computed with (\ref{HausForward}): 
starting from a slanted line segment for $\mathcal{R}^{(1)}$ as in Figure \ref{freach}, 
it can be observed that the forward reachable sets $\mathcal{R}^{(k)}$ converge rapidly, 
as confirmed with the error bound displayed in Figure \ref{littleplots} (upper plot). 

\begin{figure}[b]
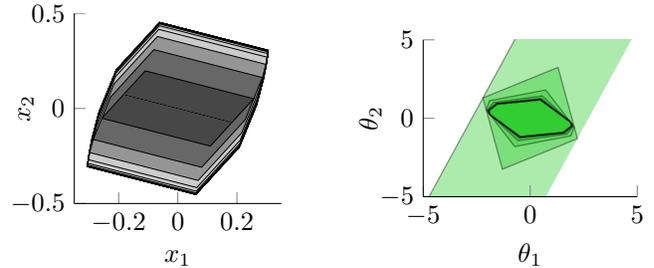

\begin{subfigure}[t]{.5\columnwidth}
 \input{freach}  \caption{The first 20 iterations of the forward reachable set $\mathcal{R}^{(k)}$, $k=1,\ldots,20$ for the case study. 
The reachable sets grow in size from dark grey ($k=1$) to light grey ($k=20$), 
so that \mbox{$\mathcal{R}^{(k-1)}\subseteq\mathcal{R}^{(k)}$}.}\label{freach}
\end{subfigure}\hfill
\begin{subfigure}[t]{.45\linewidth}
 \input{invarianceset}
\caption{The feasible sets for the $k$-bounded invariance property $\always^k\iota$, 
with $k=1,\ldots,20$, obtained for the case study.}\label{invarianceset}
\end{subfigure}  
\caption{Reachable and feasible sets for unbounded-time verification problem.}
\end{figure}
Based on $\mathcal{R}^{(k)}$, 
the feasible set for the $k$-bounded invariance $\always^k\iota$ can be computed as $\Theta_{\always^k\iota}=\Theta_{\iota}\big(\mathcal{R}^{(k)}\big)$. 
The feasible sets $\Theta_{\always^k\iota}$ with $k=1,\ldots,20$ are plotted in Figure \ref{invarianceset}. 
Observe that the feasible set $\Theta_{\always^1 \iota}$ is not bounded, but for $k\geq 2$ the feasible sets are bounded and, 
as expected, decrease in size with time. 
In Figure \ref{littleplots} (middle plot) bounds on the Hausdorff distances $\delta_H(\Theta_{\always\iota},\Theta_{\always^k\iota})$ are given for $k=2,\ldots,20$  
(no finite bound is computed for the index $k=1$, 
since for that instance the feasible set is not bounded).  
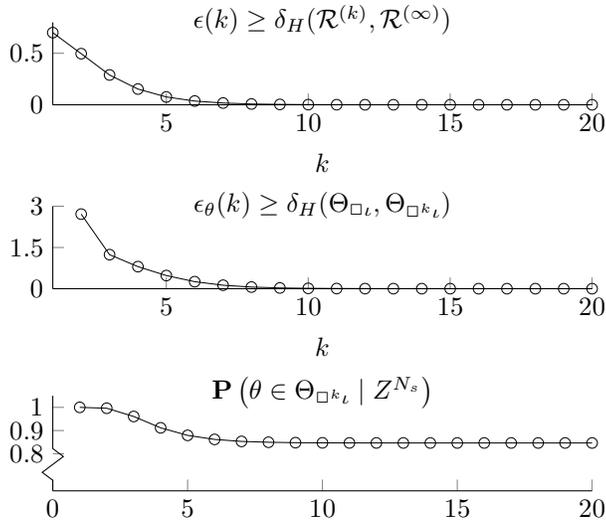
\begin{figure}[t]  {\flushright
\begin{tikzpicture}

\begin{axis}[%
width=.85\columnwidth,
height=0.13\columnwidth,
scale only axis,ytick={0,0.5,1},
xmin=1,title style ={yshift=-0.5cm},
xmax=20,
ymin=0,
ymax=0.8,xlabel={$k$},title={$\epsilon(k)\geq \delta_H(\mathcal{R}^{(k)},\mathcal{R}^{(\infty)})$},axis x line*=bottom,
axis y line*=left
]
\addplot [color=black,solid,mark=o,line width=0.3pt,mark options={solid},forget plot]
  table{1	0.700413387228011	
2	0.495998095406512	
3	0.289351218142438	
4	0.152722064061578	
5	0.0759854314748642	
6	0.0363743282993569	
7	0.0169467993818023	
8	0.00773881167895524	
9	0.00347990098165787	
10	0.00154580642544886	
11	0.000679888858562359	
12	0.000296590441508194	
13	0.000128492719848711	
14	5.53405185069649e-05	
15	2.37138401680709e-05	
16	1.01166739248384e-05	
17	4.29915265174852e-06	
18	1.82066361204476e-06	
19	7.6866960586732e-07	
20	3.23630575003009e-07	
};
\end{axis}
\end{tikzpicture}%
\begin{tikzpicture}

\begin{axis}[%
width=.85\columnwidth,
height=0.13\columnwidth,
scale only axis,ytick={0,1.5,3},xmin=1,
xmax=20,
ymin=0,
ymax=3,
xlabel={$k$},title={$\epsilon_\theta(k)\geq\delta_H(\Theta_{\always\iota},\Theta_{\always^k\iota})$},axis x line*=bottom,title style={yshift=-0.5cm},
axis y line*=left
]
\addplot [color=black,solid,mark=o,line width=0.3pt,mark options={solid},forget plot]
 table{2	2.71491664421836	
3	1.23903856161811	
4	0.8027017478536	
5	0.478287831970752	
6	0.255164278906496	
7	0.125977320134333	
8	0.0592016826995949	
9	0.0269831887729221	
10	0.0120604133185736	
11	0.0053191977543066	
12	0.00232325193852511	
13	0.00100704769556836	
14	0.000433826059764371	
15	0.000185916455465867	
16	6.00050231223236e-05	
17	3.37074164409766e-05	
18	1.06937448857388e-05	
19	4.51482953303592e-06	
20	1.91961348849067e-06	
};
\end{axis}
\end{tikzpicture}%
\begin{tikzpicture}

\begin{axis}[%
width=.85\columnwidth,
height=0.15\columnwidth,
scale only axis,
xmin=0,  ytick={0.8,0.9,1},
xmax=20,
ymin=0.64,
ymax=1.05,axis y discontinuity=crunch,axis x line*=bottom,
axis y line*=left,
title={$\p{\pa\in\Theta_{\always^k\iota}\mid Z^{N_s}}$},title style={yshift=-0.5cm}
]
\addplot[color=black,solid,mark=o,line width=0.3pt,mark options={solid},forget plot]
  table{1	0.999997792232498	
2	0.996169336112153	
3	0.960661046420747	
4	0.911651998735544	
5	0.879154739329903	
6	0.861731416647101	
7	0.853125627350337	
8	0.84945363490206	
9	0.847634055382528	
10	0.846934766776381	
11	0.846593193264184	
12	0.846525035481775	
13	0.846436419733163	
14	0.846392060024844	
15	0.846381714991712	
16	0.846378114377062	
17	0.846378114377062	
18	0.846378114377062	
19	0.846378114377062	
20	0.846380124541238	
};
\end{axis}
\end{tikzpicture}%
}%
\caption{(Upper plot) Error bound on the approximation level of the $k$-th forward reachable sets, 
which is such that $\mathcal{R}^{(\infty)} \subseteq \mathcal{R}^{(k)}+\epsilon(k) $ for $k=1,\ldots,20$.
(Middle plot) The Hausdorff distance $\epsilon_\pa(k)$ between $\Theta_{\always^k \psi}$ and $\Theta_{\always \psi}$ with $k=2,\ldots,20$, 
obtained for the case study.(Lower plot) Confidence that $\s\vDash\always^k\iota$ for $k=1,\ldots,20$ for the case in Section \ref{ex:case}, 
with a new experiment consisting of $200$ samples collected as $Z^{N_s}$. }
\label{littleplots}
\end{figure}
Let us conclude this case study looking at confidence quantification, 
as a function of the time horizon. 
Figure \ref{littleplots} (lower plot) represents the confidence over the property $\p{\pa\in \Theta_{\always^k\iota}\mid Z^{N_s}}$, for indices $k=1,\ldots,20$. 
Unlike the case discussed in Section \ref{ex:case}, 
which focused on looking at statistics of the confidence via mean and variance drawn over multiple experiments, 
we zoom in on asymptotic properties by considering a data set $Z^{N_s}$ comprising a single trace
made up of 200 measurements, 
simulated under the same conditions as in Section \ref{ex:case}, and with $\pa_0=[1\ 0]^T$.  
From the resulting probability density distribution $\pd{\pa\mid Z^{N_s}}$, 
it may be observed that the confidence converges rapidly to a nonzero value. 
 \subsection{Discussion on the Generalisation of the Results}\label{ssec:disc}
\reversemarginpar
The discussed approach based on polytopes allows for analytical expressions of the feasible set, 
however the implementation may not scale to models with very large dimension: 
in particular,   
the number of half-planes characterising the feasible set may increase with the time bound of the LTL formula $\psi$ 
(that is, with the repeated application of the $\nex$ operator), 
and with the cardinality of the atomic propositions in the alphabet $\Sigma$. 
Still, note that these computations are essentially quite similar to known reachability computations, 
therefore the method is extendable well beyond the 2-dimensional case study, 
especially when applying sophisticated reachability analysis tools in the literature. 
Therefore the discussed limitations related to the current implementation of the approach, 
ought to be dealt with in the future by the use of tailored and less na\"ive computational approaches. 

In the discussion of model selection, we hinted at possible generalisation beyond linearly-parameterised model sets. 
Future extension will deal with hybrid models,  
since when systems are not linear, their (local) behaviour is often well approximated with piecewise-linear dynamical models.

This paper has discussed the formal verification of physical systems with partly unknown dynamics, 
by introducing a Bayesian framework allowing for the efficient incorporation of measurement data and prior information within a verification procedure based on safety analysis. 
This formal approach has allowed for the computation of the confidence level over the validity of a property of interest on the unknown system. 
The method has been applied to the verification of LTI models of systems over bounded and unbounded safety properties, 
and its computational overhead has been discussed at length. 

Looking forward, 
current work targets the extension of the applicability of tractable solutions to model-based and data-driven verification over complex physical systems.   
We are presently working to extensions of the considered set of logic formulae of interest, 
and plan to employ experiment design to optimise the input-output signal interaction for efficient data usage over general classes of models, as initially attempted in \cite{ACCSofie}.
Additionally, the design of control policies that optimise properties of interest over partly unknown systems is topic of current work.

\bibliographystyle{abbrv}  
\bibliography{library}

\appendix
 
\section*{Derivation of the Bounds in Section \ref{sec:ubp}}
 
\textbf{1. Hausdorff distance of forward reachable mappings.} 
We only sketch the method to bound the Hausdorff distance, 
whereas a more formal derivation can be found in the literature on robustly positively invariant sets \cite{Blanchini2007}.     
 
The $k$-step forward reachable set equals
 \[\mathcal{R}^{(k)}:=\bigcup_{i=1}^k\left\{\sum_{j=1}^{i} A^{j-1}Bu(i-j), \ \textmd{ for } u(j)\in \A_{ver}  \right\}.\]
For $0\in \A_{ver}$, the minimal invariant set $\mathcal{R}^{\infty}$ can be written as
 \begin{align}\mathcal{R}^{(\infty)}&:=\left\{\sum_{j=0}^{i-1} A^{j}Bu(j)+A^i \sum_{k=0}^\infty A^{k}Bu(k), \textmd{ for } u(\cdot)\in \A_{ver}  \right\}.
\end{align} 
If the spectral radius of a $A$  is strictly smaller than 1, $\rho(A)<1$, then  
\begin{align}
\mathcal{R}^{(\infty)} \subseteq\mathcal{R}^{(k)} + \epsilon(k) \mathcal{B},
  \end{align}
  with 
  \[
  A^k \sum_{i=0}^\infty A^{i}Bu(k)\subseteq \epsilon(k)\mathcal{B}, \textmd{ for } u(\cdot)\in \A_{ver}.   
  \] 
 Note that $\epsilon(k)$ is bounded for $\rho(A)<1$. 
For a matrix $A$ without defective eigenvalues, i.e. where the eigenvectors form a complete basis, 
this $L_1$ norm can be easily bounded using the spectral radius of $A$, by selecting 
 \[\epsilon(k)= \frac{|\rho(A)|^k }{1-|\rho(A)|} \|B\|_2 \max_{u\in \A_{ver}}{(|u|)}\geq \|A^k\|_{2} \sum_{i=0}^\infty \|A^{i}B\|_2|u(k)|. \]
 In case that the matrix $A$ is defective, we opt to bound the $L_1$-norm by exploiting absolute sum of the $L_2$ induced norm for $A^i$ $i\rightarrow \infty$:
  \(\sum_{i=0}^\infty \|A^{i}\|_{2} \).  Note that \(\|A^{i}\|_{2}\) converges to 0 for $i\rightarrow \infty$ since $\rho(A)<1$, therefore there exists a finite $l$ such that \(\|A^{l}\|_{2}<1\) and we can upper bound the absolute sum as 
   \begin{align*}\sum_{i=0}^\infty \|A^{i}\|_{2}&\leq  \left(\sum^{l-1}_{i_1=0}  \|A^{i_1}\|_{2}\right )\left(\sum_{i_2=0}^\infty \|A^{l}\|_{2}^{i_2}\right)\\& =\left(\sum^{l-1}_{i_1=0}  \|A^{i_1}\|_{2} \right)\frac{1}{1-\|A^{l}\|_{2}} .\end{align*}
   Thus in general, the Hausdorff distance can be bounded as
    \[\delta_H(\mathcal{R}^{(k)},\mathcal{R}^{(\infty)})\leq \epsilon(k)=  \|A^k\|_{2}\max_{u\in \A_{ver}}{(|u|)}  c_1, \] 
    with $c_1=\frac{\left(\sum^l_{i_1=0}  \|A^{i_1}\|_{2} \right)}{1-\|A^{l}\|_{2}}\|B\|_2 $ for  $l$ such that $\|A^{l}\|_{2}<1$. 
    Note that $c_1$ can be replaced by any bound on the $L_1$ norm of the dynamical system formed by $(A,B)$.

In case that $\X_{ver}\not \subseteq \mathcal{R}^\infty $ then  the forward reachable  iteration can be rewritten as
\[\mathcal{R}^{(k)}_{\X_{ver}}= \left( \bigcup_{i=0}^k A^i {\X}_{ver}\right)
+\mathcal{R}^{(k)},  
\]
for which we know that
 \[\mathcal{R}^{(\infty)}_{\X_{ver}}\subseteq \mathcal{R}^{(k)}_{\X_{ver}}+\epsilon(k)+\|A\|^{k+1}  \delta_H(\X_{ver},\{0\}). 
\]
Thus the Hausdorff norm is upper bounded as\\* $\delta_H(\mathcal{R}_{\X_{ver}}^{(k)},\mathcal{R}_{\X_{ver}}^{(\infty)})\leq \epsilon(k)+\|A^{k+1} \| \delta_H(\X_{ver},\{0\})$.
 
 \bigskip
 
 \noindent\textbf{2. Hausdorff distance on feasible sets.}
 Suppose that the conditions in Lemma \ref{lem:robust} hold for $\mathcal{R}^{(k)}_{\X_{ver}}$, then we can compute a value for $\epsilon_\pa$ such that \(
 \Theta_{\psi}(\mathcal{R}^{(k)}_{\X_{ver}}) \subseteq 
\Theta_{\psi}(\mathcal{R}^{(k)}_{\X_{ver}}+\epsilon_x\mathcal{B})+\epsilon_\pa \mathcal{B}, \) 
 where $\epsilon_x$ is a bound on the Hausdorff distance $\delta_H(\mathcal{R}_{\X_{ver}}^{(k)},\mathcal{R}_{\X_{ver}}^{(\infty)})$.

The set operation $\Theta_{\psi}(\cdot)$ is monotonically decreasing, therefore 
 \(\Theta_{\psi}(\mathcal{R}^{(k)}_{\X_{ver}}+\epsilon(k)\mathcal{B})\subseteq\Theta_{\always\psi}=\Theta_{\psi}\left(\mathcal{R}^\infty_{\X_{ver}}\right) \subseteq \Theta_{\psi}\left(\mathcal{R}^{(k)}_{\X_{ver}}\right)= \Theta_{\always^k\psi},\) 
 and 
\(\Theta_{\always^k\psi} \subseteq  \Theta_{\psi}(\mathcal{R}^{(k)}_{\X_{ver}}+\epsilon(k)\mathcal{B})+\epsilon_\pa \mathcal{B} \subseteq\Theta_{\always\psi}+\epsilon_\pa\mathcal{B}, \)
and  \[\Theta_{\always\psi}\subseteq \Theta_{\always^k\psi}\subseteq\Theta_{\always\psi}+\epsilon_\pa\mathcal{B}. \]
Based on Lemma \ref{lem:robust}, with $ \epsilon_p=   \max_{p_i}\frac{|A_{p_i}|}{|b_{p_i}|  }$, we obtain 
\begin{align*}
\epsilon_\pa&=  \frac{\epsilon_x\epsilon_{p} \max_i(\|v_i\|)^2}{1+\epsilon_x\epsilon_p\max_i( \|v_i\|) }\leq \epsilon_x \epsilon_{p} \max_i(\|v_i\|)^2.
\end{align*}
Note that since $\|A^k\|_{2}$ converges to 0 for $k\rightarrow \infty$ for $\rho(A)<1$, and since $\max_i(\|v_i\|)^2$ is not increasing, the error $\epsilon_\pa$ also converges to $0$.

\end{document}